\documentclass[journal,twocolumn, 10pt, final]{IEEEtran}

\usepackage[utf8]{inputenc} 
\usepackage[T1]{fontenc}
\usepackage{url}
\usepackage{float}
\usepackage{xcolor}
\usepackage{ifthen}
\usepackage{braket}
\usepackage{tikz}
\usepackage{cite}
\usepackage{diagbox}

\usepackage[cmex10]{amsmath} % Use the [cmex10] option to ensure compliance
\usepackage{mathtools}                             
\usepackage{amsfonts}
\usepackage{amssymb}%
\usepackage{hyperref}%
\usepackage{graphicx}
\usepackage{cleveref}
\usepackage{adjustbox}
\usepackage{booktabs}
\usepackage{multirow}
\newtheorem{theorem}{Theorem}

\newtheorem{lemma}[theorem]{Lemma}

\newtheorem{proposition}[theorem]{Proposition}
\newtheorem{remark}[theorem]{Remark}

\newenvironment{proof}[1][Proof]{\noindent\textbf{#1.} }{\ \rule{0.5em}{0.5em}}

\allowdisplaybreaks

\usepackage[]{color}

\newcommand{\inv}{^{-1}}

\newcommand{\bR}{\mathbb{R}}

\newcommand{\bC}{\mathbb{C}}

\newcommand{\cS}{\mathcal{S}}

\newcommand{\cX}{\mathcal{X}}
\newcommand{\cY}{\mathcal{Y}}

\newcommand{\cM}{\mathcal{M}}

\newcommand{\cL}{\mathcal{L}}

\newcommand{\bE}{\mathbb{E}}

\newcommand{\one}{\mathbf{1}}

\newcommand{\tr}{\operatorname{tr}}

\newcommand{\eps}{\varepsilon}

\begin{document}

\title{Guesswork with Quantum Side Information}
% \title{Guesswork in the presence\\of quantum side information}
% \title{Optimal strategies for guesswork\\with quantum side information}
%\author{Nilanjana Datta \and Eric P. Hanson \and Vishal Katariya \and Mark M. Wilde}

\author{%
   \IEEEauthorblockN{Eric P. Hanson\IEEEauthorrefmark{1},  Vishal Katariya\IEEEauthorrefmark{2}, Nilanjana Datta\IEEEauthorrefmark{1},
                     and Mark M.~Wilde\IEEEauthorrefmark{2}}
    
    \IEEEauthorblockA{\IEEEauthorrefmark{1}%
                        Department of Applied Mathematics and Theoretical Physics,\\
University of Cambridge, Cambridge CB3 0WA, UK}

   \IEEEauthorblockA{\IEEEauthorrefmark{2}%
                     Hearne Institute for Theoretical Physics, Department of Physics and Astronomy, and \\Center for Computation and Technology,
                     Louisiana State University,\\Baton Rouge, Louisiana 70803, USA}
    %\IEEEauthorblockA{\IEEEauthorrefmark{3}%
                        %Department of Applied Mathematics and Theoretical Physics,
     %                   University of Cambridge,
      %                  n.datta@statslab.cam.ac.uk}
   %\IEEEauthorblockA{\IEEEauthorrefmark{4}%
                     %Hearne Institute for Theoretical Physics, Department of Physics and Astronomy, Center for Computation and Technology,\\
                     %Louisiana State University, Baton Rouge, Louisiana 70803, USA,
                     %mwilde@lsu.edu}

 }

\maketitle               
\begingroup\renewcommand\thefootnote{\textsection}
\footnotetext{This paper was presented in part at ISIT 2020.}
\footnotetext{Copyright (c) 2017 IEEE. Personal use of this material is permitted.  However, permission to use this material for any other purposes must be obtained from the IEEE by sending a request to pubs-permissions@ieee.org.}
\endgroup

\begin{abstract}
What is the minimum number of guesses needed on average to guess a realization of a random variable correctly? The answer to this question led to the introduction of a quantity called \emph{guesswork} by Massey in 1994, which can be viewed as an alternate security criterion to entropy. In this paper, we consider the guesswork in the presence of quantum side information, and show that a general sequential guessing strategy is equivalent to performing a single quantum measurement and choosing a guessing strategy based on the outcome. We use this result to deduce entropic one-shot and asymptotic bounds on the guesswork in the presence of quantum side information, and to formulate a semi-definite program (SDP) to calculate the quantity. We evaluate the guesswork for a simple example involving the BB84 states, both numerically and analytically, and we prove a continuity result that certifies the security of slightly imperfect key states when the guesswork is used as the security criterion. 
\end{abstract}
%
%
% We can remove this later; it helps me for organization

%{
%\hypersetup{linkcolor=black}
%\setcounter{tocdepth}{2}
%\tableofcontents
%}

\section{Introduction}

Information theory, among other things, concerns the security of messages against attacks by malicious agents. Conventionally, it is accepted that that the more unpredictable a message is, and the higher the (Shannon) entropy of the distribution from which it is drawn, the more secure it is to brute force attacks. Therefore, when establishing a secret key or a cipher, the gold standard is to choose a key whose elements are picked uniformly at random from some alphabet.

Entropy, however, is not the only such criterion for security. Another relevant quantity, which is also maximized by messages drawn uniformly, is the guesswork. First put forth by Massey \cite{Massey1994}, the quantity is operationally described by the following guessing game. Consider the problem of guessing a realization of a random variable $X$, taking values in a finite alphabet $\mathcal{X}$, by asking questions of the form ``Is $X=x$?''. The \emph{guesswork} $G(X)$ is defined as the minimum value of the average number of questions of this form that needs to be asked until the answer is ``yes''. That is, 
\begin{equation} \label{eq:intro-simple-guesswork}
    G(X) = \sum_{k=1}^{| \cX |} k \cdot p_G(k)
\end{equation}
where $p_G(k)$ is the probability of the $k$th guess being correct. In the real world, questions of this form arise from query access to a resource; for example, if a hacker is attempting to guess a user's password on an online portal, he or she can only ask this kind of question (as opposed to, say, ``Is $X \geq x$?'') and is allowed a limited number of guesses before being locked out. Therefore, for someone setting up a password, the number of guesses allowed by the portal provides the operational security criterion against which his or her password must compare.

In contrast, the entropy of a distribution is approximately the minimum value of the average number of guesses required to obtain the correct guess when one is allowed to ask questions of the form, ``Is $X \in \widetilde{\mathcal{X}}?$", where each $\widetilde{\mathcal{X}}$ is some subset of the alphabet $\mathcal{X}$ \cite[Theorem~5.4.1]{book1991cover}. Qualitatively speaking, entropy can be considered to be the query complexity of a binary search-type algorithm, whereas guesswork corresponds to the query complexity of a linear search-type  algorithm~\cite{Lundin2007}. Figure~\ref{fig:binary-linear-search} illustrates this difference. It is well known that binary search has a smaller complexity than linear search, which leads to the simple claim that the entropy of a distribution is always less than the guesswork. Massey \cite{Massey1994} proved the following stronger lower bound on the guesswork in terms of the Shannon entropy of the random variable $X$:
\begin{equation}
G(X) \geq \frac{1}{4} 2^{H(X)} + 1 ,
\end{equation}
provided that $H(X) \geq 2$ bits.

\begin{figure}
    \centering
    \includegraphics[width=0.48\textwidth]{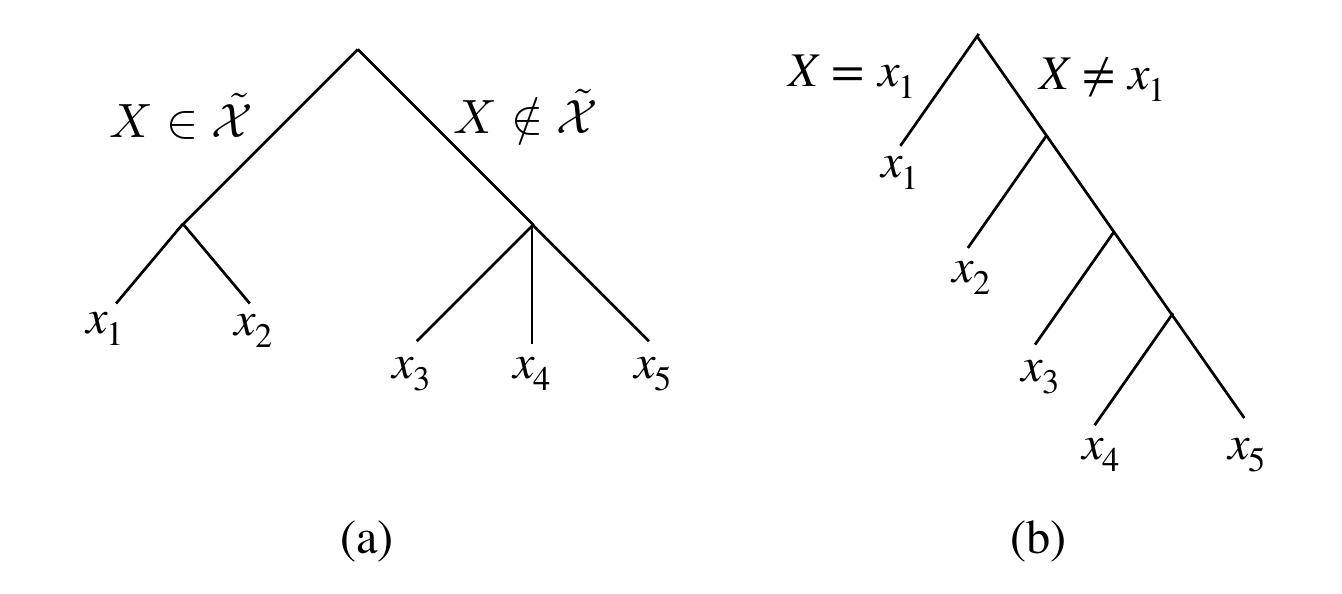}
    \caption{Consider an unknown value of $X$ in the set $\{x_1, x_2, x_3, x_4, x_5\}$. The search tree for finding the correct value of $X$ when asking questions of the form (a) ``Is $ X \in \tilde{\mathcal{X}}?$'' or of the form (b) ``Is $X = x?$'' The former search tree, similar to that in binary search, has fewer branches and the relevant operational quantity is the entropy of the prior distribution. The specific question we ask in this tree is as follows: ``Is $X \in \tilde{\mathcal{X}} = \{x_1, x_2\}$?'' The latter tree, similar to that of linear search, characterizes the scenario of guesswork. The tree encapsulates  a strategy in which one sequentially guesses $x_1, x_2$, and so on, until $x_5$.}
    \label{fig:binary-linear-search}
\end{figure}

Massey \cite{Massey1994} also showed that the optimal algorithm to minimize the guesswork, and even the positive moments of the number of guesses, consists of the intuitive strategy of simply guessing elements in decreasing order of their probability of occurrence. Guesswork can also be considered in the presence of classical side information given in the form of a random variable $Y$ that is correlated with $X$. In this case, the guesswork is the minimal number of questions of the form ``Is $X=x$?'' that is required on average to obtain the correct answer, given the value of $Y$. The average is taken over the number of guesses, with each guess number weighted by its probability of being correct. That is,
\begin{equation}
    G(X|Y=y) = \sum_{k=1}^{| \cX |} k \cdot p_G(k)
\end{equation}
where in this case, $p_G(k)$ is the probability of the $k$th guess being correct, given that $Y=y$. The optimal guessing strategy is simply to guess in decreasing order of conditional probability $p_{X|Y}(x|y)$. Arikan~\cite{Arikan1996} obtained upper and lower bounds on the guesswork and its positive moments, in this scenario as well as in the case without side information. Further work on guesswork in the classical setting has been done in~\cite{Arikan1998, Arikan1998a, Malone2004, Sundaresan2007, Hanawal2010, Christiansen2013, Sason2018a, Sason2018}.

In this paper, we consider a natural generalization of the above guessing problem to the case in which the classical side information is replaced by quantum side information. This generalization was first considered in~\cite{CCWF15}. In this case, the guesser (say, Bob) holds a quantum system $B$, instead of a classical random variable (or equivalently, a classical system) $Y$.   Here, the joint state of $X$ and $B$ is given by a classical-quantum (c-q) state, which we denote as $\rho_{XB}$ (see \Cref{sec:problem_statement} for details).
Suppose that Alice possesses the classical system $X$ containing the value of the letter $x$, which is to be guessed by Bob.
We define guesswork in the presence of quantum side information to be the minimum number of guesses needed, on average, for Bob to correctly guess Alice's choice, by performing a general sequential protocol, as follows. Bob acts on his system $B$ with a quantum instrument, yielding a classical outcome~$\hat x$ which he guesses, as well as a post-measurement state on system B. If his guess is incorrect, he performs another instrument on his system $B$ (possibly adapted based on his previous guess), and repeats the protocol until he either guesses correctly or runs out of guesses (which might be the case if he is allowed a limited number of guesses $K < |\mathcal{X}|$).

While the case of classical side information admits a very simple optimal strategy (which amounts to simply sorting the conditional probabilities $p_{X|Y}(\cdot|y)$ in non-increasing order and guessing accordingly), the quantum case requires measurement on the quantum system $B$, which potentially disturbs the state of $B$, a priori complicating the analysis of the sequence of guesses in the optimal strategy. We show in \Cref{sec:quantum_strategies}, however, that a general sequential strategy is in fact equivalent to performing a single generalized measurement yielding a classical random variable $Y$ of outcomes and then doing the optimal strategy using this $Y$ as the classical side information. The earlier work~\cite{CCWF15} instead defined guesswork with quantum side information as the latter quantity, {{i.e.}}, a measured version of the guesswork in the presence of classical side information. We prove here that these definitions are equivalent, and thus there is ultimately no difference between them, but we consider the definition in terms of a sequential protocol to be a more natural one. Moreover, the above-mentioned equivalence is proved via an explicit construction, allowing such a guessing strategy to be implemented sequentially. The single-measurement protocol could in general involve making a measurement with exponentially (in $|\mathcal{X}|$) many outcomes. Hence it may be more efficient to  implement it instead as a sequence of (linearly-many) measurements with linearly-many outcomes, as allowed by the above construction.

Moreover, we consider  a slight generalization of the guesswork in which Bob may  make only $K \leq |\cX|$ guesses in total, and in which the ``cost'' of needing to make $k$ guesses is given by a vector $\vec c = (c_1,c_2, \dotsc, c_K)$, which could be different from $(1,2, \dotsc, K)$, the latter of which corresponds to the expected value. These generalizations can be better models for certain situations; in the password-guessing example, e.g., Bob may be locked out after $K$~guesses and hence is limited to a small number of guesses, or perhaps one has to wait after each guess before making another, and the time that one waits increases with the number of incorrect guesses. We show that this generalized situation (including the guesswork as a special case) admits a semi-definite programming (SDP) representation in which the number of variables scales as $|\cX|^K$, and hence smaller values of $K$ yield smaller problems and better scaling with $|\cX|$. See \Cref{sec:SDP} for more on the computational aspects of the guesswork.

One can consider a related task, in which one wishes to maximize what is known as the ``guessing probability'' $p_{\text{guess}}(X|B)$~\cite{Koenig2009}. In this case, the guesser is given only one attempt to guess the value of $X$ (and hence is free to perform any arbitrary measurement on his system $B$). The guessing probability is related to the so-called \emph{conditional min-entropy} $H_\text{min}(X|B)$ of the c-q state $\rho_{XB}$. In some sense, we can consider the guesswork to be an extension of the guessing probability.  However, the nature of the optimization being done is different: instead of maximizing the probability of success in one attempt, we minimize the total number of guesses required. Therefore, the operations that a guesser performs to minimize the guesswork may be very different from those needed to maximize the guessing probability. Some of the connections between these two tasks have been investigated in~\cite{CCWF15}.

% \begin{remark}
% During the preparation of this manuscript, we became aware of the work of \cite{CCWF15}, which introduces a notion of guesswork with quantum side information that coincides with our definition in terms of ``measured strategies''  (see \Cref{sec:quantum_strategies}). In that paper, among other things, the authors use this formulation to obtain bounds on guesswork with quantum side information in terms of optimizations over measurements of bounds which hold in the classical case, as we do in \Cref{sec:entropic_bounds}, although they apply this technique to different bounds than we do in that Section. In that work, the authors also establish that guesswork is an SDP (their formulation coincides with our ``dual formulation'' of \Cref{sec:SDP}). We believe the rest of our results to be novel, however.
% \end{remark}

\paragraph{Overview}
In \Cref{sec:problem_statement}, we formally describe the task of guesswork with quantum side information. In \Cref{sec:strategies}, we define classical and quantum guessing strategies in a unified framework, and  \Cref{thm:ordered-from-adaptive} states that three classes of quantum strategies are equivalent. In \Cref{sec:entropic_bounds}, we establish one-shot and asymptotic entropic bounds on the guesswork, using analogous bounds developed by Arikan for the case of classical side information \cite{Arikan1996}. In \Cref{sec:SDP}, we revisit the idea that the guesswork may be formulated as a semi-definite optimization problem (SDP) (originally discussed in \cite{CCWF15}), and we use such a representation to prove that the guesswork in a concave function (see \Cref{sec:concavity}) and a Lipschitz continuous function (see \Cref{sec:Lipschitz}). We discuss the dual formulation of the SDP in \Cref{sec:dual}, a resulting algorithm to efficiently compute upper bounds in \Cref{sec:ub_algo}, and we present a mixed-integer SDP representation in \Cref{sec:misdp}. \Cref{sec:example} shows two simple examples of the guesswork involving the BB84 states and the Y states, and \Cref{sec:security} provides a robustness result for using guesswork as a security criterion.

\section{Statement of the problem}\label{sec:problem_statement}

Alice chooses a letter $x\in \cX$ with some probability $p_X(x)$, where $\cX$ is a finite alphabet. This naturally defines a random variable $X\sim p_X(x)$. She then sends a quantum system $B$ to Bob, prepared in the state  $\rho_B^{x}$, which depends on her choice $x$. Bob knows the set of states $\{\rho_B^x : x \in \cX\}$, and the probability distribution $\{p_X(x) : x \in \cX\}$, but he does not know which particular state is sent to him by Alice. Bob's task is to guess $x$ correctly with as few guesses as possible.
From Bob's perspective, he therefore has access to the $B$-part of the c-q state
\begin{equation}\label{eq:cq-state}
\rho_{XB} := \sum_x p_X(x)|x\rangle \!\langle x|_X \otimes \rho_B^x.
\end{equation}
% I don't think we need this:
% Without loss of generality, we may assume the reduced state $\rho_B := \tr_X[\rho_{XB}]$ is full rank; otherwise, one may simply reduce the size of the Hilbert space corresponding to the system $B$ to the support of $\rho_B$.

In the purely classical case, this task reduces to the following scenario: Alice holds the random variable $X\sim p_X(x)$, and Bob holds a correlated random variable $Y$ and knows the joint distribution of $(X,Y)$. In this case $\rho_{XB}$ reduces to the state
\begin{equation} \label{eq:cq-classical}
\rho_{XY} = \sum_x p_X(x)|x\rangle \!\langle x|_X \otimes \sum_y p_{Y|X}(y|x) |y\rangle \!\langle y|_Y.
\end{equation}
In this case, if Bob's random variable $Y$ has value $y$, then an optimal guessing strategy is to sort the conditional distribution $p_{X|Y}(\cdot |y)$ in non-increasing order so that
\begin{equation}
p_{X|Y}(x_1 | y) \geq p_{X|Y}(x_2 | y) \geq \dotsc \geq p_{X|Y}(x_{|\cX|} | y)
\end{equation}
and simply guess first $x_1$, then $x_2$, etc., until he gets it correct~\cite{Arikan1996}. We note that in the absence of side information, Bob simply guesses in non-increasing order of $p_X(\cdot)$.

In the case in which Bob's system $B$ is quantum, he is allowed to perform any local operations he wishes on $B$, and then make a first guess $x_1$. He is told by Alice whether or not his guess is correct; then he can perform local operations on $B$, and make another guess, and so forth. We are interested in determining the minimal number of guesses needed on average for a given ensemble $\{ p_X(x), \rho_B^x \}_{x\in \cX}$ and the associated optimal strategy.

More generally, we allow Bob to make $K$ guesses, with possibly $K < |\cX|$. Formally, we assume that Bob always makes all $K$ guesses; any guess after the correct guess simply does not factor into the calculation of the minimal number of guesses (see \Cref{sec:strategies} for a more detailed definition of the minimal number of guesses). Thus, Bob makes a sequence of guesses, $g_1,\dotsc,g_K \in \cX^K_{\neq}$ with some probability. 

We could consider the scenario in which Bob makes a guess $x_1$, then learns whether or not the guess was correct, and uses that information to make his second guess $x_2$, and so forth. However, if Bob learns that his $j$th guess $x_j$ is correct, then it does not matter what he guesses subsequently (it has no bearing on the minimal number of guesses). If the guess is incorrect, then his subsequent guesses do matter, and he should make his next guess accordingly. Hence, in such a protocol, the feedback about whether or not the $j$th guess is correct does not help, and Bob might as well assume that each guess is incorrect.

% We therefore let $\vec G = (G_1,\dotsc,G_K)$ be the joint random variable representing Bob's sequence of guesses. For example, $G_1$ is the random variable corresponding to Bob's first guess.

\section{Guessing strategies} \label{sec:strategies}

When Alice chooses $x^* \in \cX$, a guessing strategy for Bob outputs a sequence of guesses $\vec g = (g_1,\dotsc,g_K) \in \cX^K$ with some probability $p_{\vec G | X}(\vec g | x^*)$. Hence, formally, a \emph{guessing strategy for $X$ with $K$ guesses} is a random variable $\vec G$ on $\cX^K$ that is correlated with $X$, such that the joint random variable $(X, \vec G)$ has marginal $X \sim p_X$. Note that the definition of a guessing strategy has no reference to the side information (if any) to which Bob has access; instead, the side information dictates the set of guessing strategies to which Bob has access. This allows various types of side information to be analyzed within a uniform framework; in particular, the set of strategies available when Bob has access to some classical side information $Y$ is described in \Cref{sec:classical_strategies}, while the case of quantum side information is described in \Cref{sec:quantum_strategies}. 

We are interested in the minimal number of guesses required to guess $x^*$ correctly with a fixed sequence of guesses $\vec g$. This is defined as follows:
\begin{multline}\label{eq:def_N}
N(\vec g, x^*) := \\
\begin{cases}
\min \left\{ j :  g_j = x^*\right\} & g_j = x^* \text{ for some }j =1,\dotsc, K\\
\infty & \text{else},
\end{cases}
\end{multline}
where the outcome $\infty$ occurs when none of the $K$ guesses are correct. We can view $N$ as a random variable taking values in $\{1,2,\dotsc,K, \infty\}$. Given a guessing strategy $\vec G$, the quantity of interest is $N(\vec G, X)$, the corresponding random variable. We define
\begin{equation}
\cS_K(X) := \left\{ N(\vec G, X) : X\vec{G} \sim p_{X\vec{G}} \right\}
\end{equation}
to be the set of all possible random variables $N$ associated to all guessing strategies $\vec G$ with $K$ guesses.
We say two guessing strategies $\vec G$ and $\vec G'$ for $X$ with $K$ guesses are \emph{equivalent} if $N(\vec G', X) = N(\vec G, X)$.

Note that if $\vec G$ and $\vec G'$ are two strategies with $K$ guesses for $X$ that  differ only in guesses made after guessing the correct answer, then they are equivalent. This formalizes the notion introduced at the end of the previous section: since guesses made after the correct answer do not change the value of $N(\vec g, x^*)$, feedback of whether or not $g_j = x^*$ can only lead to equivalent strategies.

\subsection{Classical strategies} \label{sec:classical_strategies}

Consider a pair of random variables $(X,Y)$ where $X$ has a finite alphabet $\cX$ and $Y$ has a countable alphabet $\cY$. Alice chooses $x^* \in \cX$ (with probability $p_X(x^*)$) and Bob is given $y\in \cY$ (with probability $p_{Y|X}(y|x^*)$). Bob's task is to guess $x^*$. Since Bob's sequence of guesses $(g_1,\dotsc,g_K)$ can only depend on $x^*$ via $y$, a classical guessing strategy $\vec G$ is any random variable $\vec G$ such that the ordered triple $(X,Y,\vec G)$ of random variables  forms a Markov chain, which we denote as $X - Y - \vec G$. Hence, given a joint probability distribution $p_{XY}$, we define the set of random variables $N$ associated to classical guessing strategies as follows:
\begin{equation}
\cS_K^\textnormal{Classical}(p_{XY}) := \left\{ N(\vec G, X) : X - Y - \vec G \right\} \subseteq \cS_K(X).
\end{equation}

\subsection{Equivalence of quantum strategies} \label{sec:quantum_strategies}

Let us consider three classes of quantum strategies:
\begin{enumerate}
\item Measured strategy: Bob performs an arbitrary POVM $\{E_y\}_{y \in \cY}$ on the $B$ system. Let $Y$ be the random variable with outcomes in a finite alphabet $\cY$ corresponding to his measurement outcomes, {{i.e.}}
\begin{equation}\label{eq:Y-from-X}
  p_{Y|X}(y|x) = \tr[E_{y} \rho_B^x], \quad \forall x\in \cX,\, y \in \cY.
  \end{equation}
% The pair of random variables $(X,Y)$ can be embedded in the following quantum state:
% \begin{equation}
% \sigma_{XY} = \sum_{x\in \cX} p_X(x) |x\rangle \!\langle x| \otimes \sum_{y \in \cY} \tr[E_{y} \rho_B^x] |y\rangle \!\langle y|.
% \end{equation}
 Bob then employs a classical guessing strategy on $(X,Y)$. The set of random variables corresponding to the possible number of guesses under such a strategy is given by

\begin{multline}
\cS_K^\textnormal{Measured}(\rho_{XB}) := \Big\{ N(\vec G, X) : X - Y - \vec G, \\ \text{ $Y$ satisfies \eqref{eq:Y-from-X} for finite alphabet } \cY \\  \text{\& POVM } \{E_y\}_{y \in \cY} \Big\}.
\end{multline}

\iffalse
\begin{equation}
    \cS_K^\textnormal{Measured}(\rho_{XB}) := \Big\{ N(\vec G, X) : X - Y - \vec G, \text{ $Y$ satisfies \eqref{eq:Y-from-X} for a finite alphabet $\cY$ \& POVM } \{E_y\}_{y \in \cY} \Big\}.
\end{equation}
\fi
We then observe that
\begin{equation}
\cS_K^\textnormal{Measured}(\rho_{XB}) \subseteq \cS_K(X).    
\end{equation}

% 	The minimal expected number of guesses under such a strategy is given by
% \[
% E_N^\textnormal{Measured}(\rho_{XB}) := \inf_{p_{XY} = \id_X\otimes\cN_{B\to Y}(\rho_{XB})}\inf_G(E_N(p_{XY}, G)),
% \]
% where  $E_N(p_{XY}, G)$ denotes the expected number of guesses for the guessing problem when $p_{XY}$ is the joint distribution and $G$ is a classical guessing strategy.
\item Ordered strategy: Bob performs a measurement with outcomes in $\cX^K$, which are identified with guessing orders; {{i.e.}}, if the outcome is $(x_1,\dotsc,x_K) \in \cX^K$, Bob first guesses $x_1$, then $x_2$, and so forth. In this case, Bob performs a POVM $\{E_{\vec g}\}_{\vec g \in \cX^K}$ and the guessing strategy $\vec G$ is distributed according to
\begin{equation}\label{eq:G-dist-ordered}
p_{\vec G | X}(\vec g | x) =  \tr[E_{\vec g} \rho_B^{x}].
\end{equation}
As above, we define
\begin{multline}
\cS_K^\textnormal{Ordered}(\rho_{XB}) := \Big\{ N(\vec G, X) : (\vec G, X) \text{ satisfy} \\  \text{\eqref{eq:G-dist-ordered} for some POVM } \{E_{\vec g}\}_{\vec g \in \cX^K}\Big\}
\subseteq \cS_K(X).
\end{multline}
\iffalse
\begin{equation}
    \cS_K^\textnormal{Ordered}(\rho_{XB}) := \Big\{ N(\vec G, X) : (\vec G, X) \text{ satisfy \eqref{eq:G-dist-ordered} for some POVM } \{E_{\vec g}\}_{\vec g \in \cX^K}\Big\}
\end{equation}
\fi
It is evident that
\begin{equation}
\cS_K^\textnormal{Ordered}(\rho_{XB}) \subseteq \cS_K^\textnormal{Measured}(\rho_{XB})
\end{equation}
because any such ordered strategy is a special type of measured strategy (with $Y = \vec G$). However, any measured strategy can in fact be simulated by an ordered strategy. Suppose we have a measured strategy with alphabet~$\cY$, POVM $\{E_y\}_{y\in \cY}$, and $\vec G$ satisfying $X-Y-\vec G$. Then 
\begin{align} \label{eq:meas_by_ordered_step}
p_{\vec G|X}( \vec g|x) &= \sum_{y \in \cY} p_{\vec G|Y}( \vec g|y) p_{Y|X}(y | x) \\
&=  \sum_{y \in \cY} p_{\vec G| Y}(\vec g|y) \tr[E_{y} \rho_B^x],
\end{align}
where we have used the Markov property for the first equality and \eqref{eq:Y-from-X} for the second equality. 

Let $\tilde E_{\vec g} := \sum_{y\in \cY} p_{\vec G|Y}(\vec g|y) E_{y}$. Note $\{\tilde E_{\vec g}\}_{\vec g \in \cX^K}$ is a POVM: each element is positive semi-definite since $\{E_y\}_{y \in \cY}$ is a POVM, and 
\begin{align}
\sum_{\vec g \in \cX^K} E_{\vec g} &=\sum_{\vec g \in \cX^K}  \sum_{y\in \cY}  p_{\vec G|Y}(\vec g|y) E_{y} \\ 
&=   \sum_{y\in \cY}  \sum_{\vec g \in \cX^K}p_{\vec G|Y}(\vec g|y)E_{y} = \sum_{y \in \cY} E_y = \one_B,
\end{align}
using again that $\{E_y\}_{y \in \cY}$ is a POVM.
Then substituting the definition of $\tilde E_{\vec g}$ into \eqref{eq:meas_by_ordered_step} yields
\begin{equation}
p_{\vec G|X}(\vec g|x) =  \tr[\tilde E_{\vec g} \rho_B^x]
\end{equation}
and hence \eqref{eq:G-dist-ordered} is satisfied with $E = \tilde E$. Therefore,
\begin{equation}
\cS^\textnormal{Ordered}(\rho_{XB}) = \cS^\textnormal{Measured}(\rho_{XB}).
\label{eq:ordered-=-measured}
\end{equation}

% Letting $E_N^\textnormal{Ordered}(\rho_{XB})$ be the optimal expected number of guesses under such a strategy, we immediately have
% \[
% E_N^\textnormal{Ordered}(\rho_{XB}) \geq E_N^\textnormal{Measured}(\rho_{XB}) 
% \]
% since this is a particular type of measured strategy.
\item \label{it:sequential} Sequential quantum strategy:  
Suppose that Alice chooses $x$ (which occurs with probability $p_X(x)$), and hence Bob has the state $\rho_B^x$. To make his first guess, Bob chooses a set of generalized measurement operators $\{ M_x^{(1)}\}_{x\in \cX}$ and reports the measurement outcome as his guess. He gets outcome $x_1$ with probability
\begin{equation}
p_{G_1|X}(x_1|x)= \tr[ M_{x_1}^{(1)} \rho_B^x  M_{x_1}^{(1)}{}^\dagger]
\end{equation}
and his post-measurement state is
\begin{equation}
\frac{1}{p_{G_1|X}(x_1|x)} M_{x_1}^{(1)} \rho_B^x  M_{x_1}^{(1)}{}^\dagger.
\end{equation}

{\em{Note:}} in general, Bob could perform a unitary operation $U_1$ on his state before measuring it. However, this would simply correspond to measuring with $\{M_{x}^{(1)} U_1\}_{x \in \cX}$ instead. Hence, it suffices to simply consider a generalized measurement $\{ M_x^{(1)}\}_{x\in \cX}$.

Then, after learning the outcome $x_1$, Bob chooses a new set of generalized measurement operators $\{ M_x^{(2 | x_1)}\}_{x\in \cX}$. Note that this set of measurement operators can depend on $x_1$. Without loss of generality, we can keep the same outcome set $\cX$, since Bob could set, {{e.g.}}~$M_{x_1}^{(2 | x_1)} = 0$ to avoid guessing the same number twice. Bob measures his state and gets the outcome $x_2$ with probability
\begin{multline}
p_{G_2|G_1 X}(x_2|x_1, x) \\
= \frac{1}{p_{G_1|X}(x_1|x)} \tr[M_{x_2}^{(2| x_1)} M_{x_1}^{(1)} \rho_B^x  M_{x_1}^{(1)}{}^\dagger M_{x_2}^{(2| x_1)}{}^\dagger].
\end{multline}
Multiplying by $p_{G_1|X}(x_1|x)$ we see the joint distribution is given by
\begin{multline}
p_{G_1 G_2|X}(x_1, x_2|x)  \\ 
= \tr[M_{x_2}^{(2| x_1)} M_{x_1}^{(1)} \rho_B^x  M_{x_1}^{(1)}{}^\dagger M_{x_2}^{(2| x_1)}{}^\dagger].
\end{multline}
To make his $j$th guess, we allow Bob to choose a new set of generalized measurement operators\\$\{M_x^{(j|x_1,\dotsc,x_{j-1})}\}_{x\in \cX}$,  which may depend on the previous $j-1$ outcomes. Repeating the previous logic, in the end we find that
%\iffalse
\begin{multline} \label{eq:general_prob_G}
p_{G_1 G_2 \dotsm G_K  |X}(x_1, x_2,\dotsc, x_K|x) = \\
\tr[M_{x_K}^{(K| x_1, x_2,\dotsc,x_{K-1})}  \dotsm M_{x_2}^{(2| x_1)} M_{x_1}^{(1)} \rho_B^x \\
M_{x_1}^{(1)}{}^\dagger M_{x_2}^{(2| x_1)}{}^\dagger\dotsm M_{x_K}^{(K| x_1, x_2,\dotsc,x_{K-1})}{}^\dagger].
\end{multline}
%\fi
%\begin{equation} \label{eq:general_prob_G}
%    p_{G_1 G_2 \dotsm G_K  |X}(x_1, x_2,\dotsc, x_K|x) = \tr[M_{x_K}^{(K| x_1, x_2,\dotsc,x_{K-1})}  \dotsm M_{x_2}^{(2| x_1)} M_{x_1}^{(1)} \rho_B^x \\  M_{x_1}^{(1)}{}^\dagger M_{x_2}^{(2| x_1)}{}^\dagger\dotsm M_{x_K}^{(K| x_1, x_2,\dotsc,x_{K-1})}{}^\dagger].
%\end{equation}
Under such a strategy, the possible random variables giving the number of guesses is given by
\begin{multline}
\cS^{\textnormal{Sequential}}(\rho_{XB}) := \Big\{ N(\vec G, X) : \\
 (\vec G, X) \text{ satisfy \eqref{eq:general_prob_G} for some collections of} \\ 
\text{measurement operators } \{M_{x_j}^{(j| x_1, x_2,\dotsc,x_{j-1})}\}_{x_j \in \cX},\\
\, j =1,\dotsc,K,\,\, x_1,x_2,\dotsc, x_K \in \cX \Big\}.
\end{multline}
\iffalse
\begin{multline}
    \cS^{\textnormal{Sequential}}(\rho_{XB}) := \Big\{ N(\vec G, X) : (\vec G, X) \text{ satisfy \eqref{eq:general_prob_G} for some collections of measurement operators}  \\
    \{M_{x_j}^{(j| x_1, x_2,\dotsc,x_{j-1})}\}_{x_j \in \cX}, j =1,\dotsc,K,\,\, x_1,x_2,\dotsc, x_K \in \cX \Big\}.
\end{multline}
\fi

% Suppose Alice holds $x^*$. Bob can do any local operations he wishes, then outputs a guess $x_1$. He then performs further local operations which may be conditioned on $x_1$, and make a subsequent guess $x_2$. He repeats this until he has made $K$ guesses.

% Let $E_N^*(\rho_{XB})$ denote the optimal expected number of guesses under all such strategies. Then
% \[
% E_N^*(\rho_{XB}) \leq E_N^\textnormal{Measured}(\rho_{XB}) \leq E_N^\textnormal{Ordered}(\rho_{XB}).
% \]
\end{enumerate}

\begin{theorem}\label{thm:ordered-from-adaptive}
Let $\rho_{XB}$ be a c-q state as defined in \eqref{eq:cq-state} and $K$ a natural number with $K \leq |\mathcal{X}|$. Then
\begin{equation}\label{eq:all-q-strat-equiv}
 \cS^{\textnormal{Sequential}}_K(\rho_{XB}) = \cS^{\textnormal{Ordered}}_K(\rho_{XB}) = \cS^{\textnormal{Measured}}_K(\rho_{XB}).
\end{equation}
\end{theorem}
We see that all three sets of random variables of the number of guesses obtained from various classes of strategies all coincide. Hence, we call the single class that of {\em{quantum strategies}}, denoted as $\cS^\textnormal{Quantum}_K(\rho_{XB})$.

\begin{proof}
The second equality  was already stated in \eqref{eq:ordered-=-measured} and proven before that, and so it remains to prove the first equality.
Consider a sequential strategy, with the notation of point \ref{it:sequential} above.
Define
\begin{multline} \label{eq:E_from_M}
E_{x_1,\dotsc,x_K} :=  M_{x_1}^{(1)}{}^\dagger M_{x_2}^{(2| x_1)}{}^\dagger\dotsm M_{x_K}^{(K| x_1, x_2,\dotsc,x_{K-1})}{}^\dagger  \\ 
M_{x_K}^{(K| x_1, x_2,\dotsc,x_{K-1})}  \dotsm M_{x_2}^{(2| x_1)} M_{x_1}^{(1)} .
\end{multline}

We see that $E_{x_1,\dotsc,x_K} = A^\dagger A$ for $A=M_{x_K}^{(K| x_1, x_2,\dotsc,x_{K-1})}  \dotsm M_{x_2}^{(2| x_1)} M_{x_1}^{(1)}$, and hence it is positive semi-definite. Moreover,
\begin{align}	
\sum_{x_1,\dotsc,x_K \in \cX} E_{x_1,\dotsc,x_K} = I_{B}
\end{align}
as can be seen by first summing \eqref{eq:E_from_M} over $x_K$, using 
\begin{equation}
\sum_{x_K \in \cX} M_{x_K}^{(K| x_1, x_2,\dotsc,x_{K-1})}{}^\dagger M_{x_K}^{(K| x_1, x_2,\dotsc,x_{K-1})}  = I_B
\end{equation}
since $\{M_{x}^{(K| x_1, x_2,\dotsc,x_{K-1})}\}_{x\in \cX}$ is a set of generalized measurement operators, and then similarly summing over $x_{K-1}$, $x_{K-2}$,\ldots, and finally $x_1$. Let us write $E_{\vec x}$ where $\vec x=(x_1,\dotsc,x_K)$ for $E_{x_1,
\dotsc, x_K}$. We have shown that $\{ E_{\vec x} \}_{\vec x \in \cX^K}$ is a POVM.
 Moreover,
\begin{equation}
p_{G_1 G_2 \dotsm G_K  |X}(x_1, x_2 ,\dotsc, x_K|x) = \tr[E_{x_1,\dotsc,x_K} \rho_B^x].
\end{equation}
Hence, Bob's strategy is equivalent to simply performing the single POVM $\{ E_{\vec x} \}_{\vec x \in \cX^K}$ once, obtaining an outcome $\vec x = (x_1,\dotsc, x_K)$, and then making $x_1$ his first guess, $x_2$ his second guess, and so forth. That is, any such strategy can be recast as an ordered strategy.

 On the other hand, any such ordered strategy can be reformulated as an adaptive strategy, by the following recursive approach. Suppose that we are given $\left\{ E_{\vec y} \right\}_{\vec y \in \cX^K}$.
 For each $x_1 \in \cX$, define
 \begin{equation}
 M^{(1)}_{x_1} = \sqrt{\sum_{x_2,\dotsc,x_K\in \cX} E_{x_1,\dotsc,x_K}}
 \end{equation}
 where we have chosen the positive semi-definite square root. We have that
 \begin{align}
 \sum_{x_1\in \cX} M^{(1)}_{x_1}{}^\dagger M^{(1)}_{x_1} &=  \sum_{x_1 \in \cX} (M^{(1)}_{x_1})^2 \\
 &= \sum_{x_1\in \cX} \sum_{x_2,\dotsc,x_K \in \cX} E_{x_1,\dotsc,x_K} = I_B,
 \end{align}
 and so $\left\{  M^{(1)}_{x_1} \right\}_{x_1 \in \cX}$ forms a set of generalized measurement operators with outcomes in $\cX$. Next, for each $x_1 \in \cX$, corresponding to obtaining outcome $x_1$ on the first measurement, we define measurement operators $\{M^{(2| x_1)}_{x_2}\}_{x_2 \in \cX}$ by
 \begin{equation}
 M^{(2| x_1)}_{x_2} = \sqrt{(M^{(1)}_{x_1})\inv\sum_{x_3,\dotsc,x_K \in \cX} E_{x_1,\dotsc,x_K}(M^{(1)}_{x_1})\inv}.
 \end{equation}
 Then
\begin{align}	
 &\sum_{x_2\in \cX} (M^{(2| x_1)}_{x_2})^2 \\ 
 & \qquad = (M^{(1)}_{x_1})\inv\sum_{x_2\in \cX}\sum_{x_3,\dotsc,x_K\in \cX} E_{x_1,\dotsc,x_K}(M^{(1)}_{x_1})\inv \\
 & \qquad = (M^{(1)}_{x_1})\inv (M^{(1)}_{x_1})^2(M^{(1)}_{x_1})\inv \\
 & \qquad = I_B.
\end{align}
Likewise, we define
\begin{multline}
M^{(3| x_1, x_2)}_{x_3} = \Big\{ (M^{(2| x_1)}_{x_2})\inv(M^{(1)}_{x_1})\inv\sum_{x_4,\dotsc,x_K\in \cX} \\ E_{x_1,\dotsc,x_K}(M^{(1)}_{x_1})\inv(M^{(2| x_1)}_{x_2})\inv \Big\}^{1/2}.
\end{multline}

Then
\begin{align}
 & \sum_{x_3\in \cX} (M^{(3| x_1, x_2)}_{x_3})^2 \\
 &  = \begin{multlined}[t][5cm] (M^{(2| x_1)}_{x_2})\inv(M^{(1)}_{x_1})\inv \left( \sum_{x_3\in \cX}\sum_{x_4,\dotsc,x_K\in \cX} E_{x_1,\dotsc,x_K} \right) \\ (M^{(1)}_{x_1})\inv(M^{(2| x_1)}_{x_2})\inv  \end{multlined} \\
 & = (M^{(2| x_1)}_{x_2})\inv (M^{(2| x_1)}_{x_2})^2 (M^{(2| x_1)}_{x_2})\inv \\
 &= I_B.
\end{align}
Repeating this process, we define
\begin{multline}	
  M^{(\ell| x_1, x_2, \dotsc, x_{\ell-1})}_{x_\ell} \\
 =  \Big\{(M^{(\ell-1| x_1, x_2, \dotsc, x_{\ell-2})}_{x_{\ell-1}})\inv\dotsm (M^{(1)}_{x_1})\inv \\ 
 \sum_{x_{\ell+1},\dotsc,x_K\in\cX} E_{x_1,\dotsc,x_K} \\ 
 (M^{(1)}_{x_1})\inv\dotsm M^{(\ell-1| x_1, x_2, \dotsc, x_{\ell-2})}_{x_{\ell-1}})\inv \Big\}^{1/2}
\end{multline}
\iffalse
\begin{multline}
    M^{(\ell| x_1, x_2, \dotsc, x_{\ell-1})}_{x_\ell} \\
 =  \sqrt{(M^{(\ell-1| x_1, x_2, \dotsc, x_\ell-2)}_{x_{\ell-1}})\inv\dotsm (M^{(1)}_{x_1})\inv \sum_{x_{\ell+1},\dotsc,x_K\in\cX} E_{x_1,\dotsc,x_K}(M^{(1)}_{x_1})\inv\dotsm M^{(\ell-1| x_1, x_2, \dotsc, x_\ell-2)}_{x_{\ell-1}})\inv }
\end{multline}
\fi
to obtain a generalized measurement operator for step $\ell$ (to use when having obtained outcomes $x_1,\dotsc, x_{\ell-1}$ during the previous steps).
At the last step, $\ell = K$, there is no sum, namely 
\begin{multline}	
  M^{(K| x_1, x_2, \dotsc, x_{K-1})}_{x_K} \\
 =  \Big\{(M^{(K-1| x_1, x_2, \dotsc, x_K-2)}_{x_{K-1}})\inv\dotsm (M^{(1)}_{x_1})\inv E_{x_1,\dotsc,x_K}  \\
 (M^{(1)}_{x_1})\inv\dotsm M^{(K-1| x_1, x_2, \dotsc, x_K-2)}_{x_{K-1}})\inv\Big\}^{1/2}.
\end{multline}
\iffalse
\begin{equation}
    M^{(K| x_1, x_2, \dotsc, x_{K-1})}_{x_K} =  \sqrt{(M^{(K-1| x_1, x_2, \dotsc, x_K-2)}_{x_{K-1}})\inv\dotsm (M^{(1)}_{x_1})\inv E_{x_1,\dotsc,x_K}   (M^{(1)}_{x_1})\inv\dotsm M^{(K-1| x_1, x_2, \dotsc, x_K-2)}_{x_{K-1}})\inv }.
\end{equation}
\fi
Lastly, we check that by design, \eqref{eq:E_from_M} holds. Thus, we can work backwards from that equation and see that our newly created adaptive strategy yields the same outcomes with the same probabilities as the initial ordered strategy.
\end{proof}

\subsection{Success metrics}

Given a random variable $X$ and a maximal number $K$ of allowed guesses, how do we measure the success of a guessing strategy $\vec G$? We will focus on expectations of $N(\vec G, X)$. In particular, we consider the expected number of guesses required to guess correctly:
\begin{equation} \label{eq:expected_guess_inf}
\bE[N(\vec G, X)] = \begin{cases} \sum_{k = 1}^K k\cdot p_{N(\vec G, X)}(k) & \!\!\!\text{if } p_{N(\vec G, X)}(\infty) = 0\\
\infty & \!\!\!\text{if } p_{N(\vec G, X)}(\infty) > 0.
\end{cases}
\end{equation}
Here, $p_{N(\vec G, X)}(k)$ is the probability that, for guessing strategy $\vec G$, the $k$th guess is correct.
We also consider a general cost vector $\vec c$ $ = \{ c_1, c_2, \dotsm ,c_{| \cX |}\} $ $\in (\bR\cup \{\infty\})^{|\cX|}$  with
\begin{equation}\label{eq:QSI_cost_vector}
0 \leq c_1 \leq c_2 \leq \dotsm \leq c_{|\cX|}.
\end{equation}
Then we define the modified expectation
\[
E_{\vec c}(N(\vec G, X)) := \sum_{k = 1}^{|\cX|} c_k\cdot p_{N(\vec G, X)}(k).
\]
Imposing a maximal number $K < |\cX|$ of allowed guesses is equivalent to choosing $c_{K+1} = \dotsm = c_{|\cX|} = \infty$, using the convention $\infty \cdot 0 = 0$. Accordingly, we implicitly associate $K$ with $\vec c$ in all the following via the rule that $K = |\cX|$ if and only if $c_{|\cX|}<\infty$, and otherwise $K=\min \{ i : c_i = \infty \}$. The case $K=|\cX|$ therefore corresponds to $|\cX|$ guesses being allowed, each with finite cost, and the case $K < |\cX|$ corresponds to a limited number of allowed guesses, with a corresponding infinite cost if the correct answer is not obtained in $K$ guesses.

Given a c-q state $\rho_{XB}$ and a cost vector $\vec c$ as in \eqref{eq:QSI_cost_vector}, we define the generalized guesswork with quantum side information as
\begin{equation}\label{eq:def_Ec-quantum}
G_{\vec c}(X|B)_\rho := \inf_{N \in \cS_K^\textnormal{Quantum}(\rho_{XB})} E_{\vec c}(N).
\end{equation}
Likewise, given a joint distribution $p_{XY}$, let
\begin{equation}
G_{\vec c}(X|Y)_p := \inf_{N \in \cS_K^\textnormal{Classical}(p_{XY})} E_{\vec c}(N).
\end{equation}
From the equality
\begin{equation}
\cS_K^\textnormal{Quantum}(\rho_{XB}, K) = \cS_K^\textnormal{Measured}(\rho_{XB}) 
\end{equation}
of \Cref{thm:ordered-from-adaptive} it follows that
\begin{equation} \label{eq:Ec_q_from_c}
 G_{\vec c}(X|B)_\rho = \inf_{\{E_y\}_{y \in \cY}} G_{\vec c}(X|Y)_p
\end{equation}
where the infimum is over all finite alphabets $\cY$ and POVMs  $\{E_y\}_{y \in \cY}$ and $p_{XY}(x, y) = p_X(x) \tr[E_y \rho_B^x]$.

In the standard case in which $\vec c = (1,2,\dotsc,|\cX|)$, we define the \emph{guesswork with quantum side information} as
\begin{equation} \label{eq:def_G}
G(X|B) \equiv G(X|B)_\rho := G_{\vec c}(X|B)_\rho
\end{equation}
and likewise define $G(X|Y)_p = G_{\vec c}(X|Y)_p$ in the case of classical side information $Y$.

\begin{remark}\label{rem:rank_1}
In Ref.~\cite{CCWF15}, guesswork with quantum side information was defined by the right-hand side of \eqref{eq:Ec_q_from_c} (with $\vec c = (1,2,\dotsc,|\cX|)$). Moreover, Proposition 1 of that work shows that the infimum in \eqref{eq:Ec_q_from_c} in that case may be restricted to POVMs whose elements are all rank~one.
\end{remark}

% \subsection{A simple suboptimal strategy}
% We may consider a simple quantum strategy by first

% repeatedly choosing the measurement operators to maximize the probability of the next guess being correct. This is a ``greedy'' strategy, in the sense that while at each step it maximizes the probability of the next step being correct, the postmeasurement state may be severl

\section{Entropic bounds} \label{sec:entropic_bounds}

In this section, we use the results of \Cref{sec:strategies} to obtain one-shot and asymptotic entropic bounds on $G(X|B)$ in terms of measured versions of bounds known in the classical case. 

\subsection{One-shot bounds}

% Notation: in the classical case, let $E_N(p_{XY}, G)$ denote the expected number of guesses for the guessing problem when $p_{XY}$ is the joint distribution and $G$ is a classical guessing strategy.

% Likewise, let $E_N(\rho_{XB}, \{E_{\vec y}\}_{\vec y})$ be the expected number of guesses with a single-POVM strategy when the joint c-q state is $\rho_{XB}$ and $\{E_{\vec y}\}_{\vec y}$ is the single POVM.

In the case in which $K = |\cX|$, Arikan~\cite{Arikan1996} showed that 
\begin{equation} \label{eq:Arikan-1shot-bound}
\frac{1}{1 + \ln|\cX|} \exp(H_{\frac{1}{2}}^\uparrow(X|Y)_p) \leq G(X|Y)_p \leq \exp(H_{\frac{1}{2}}^\uparrow(X|Y)_p)
\end{equation}
where $H_{\alpha}^\uparrow(X|Y)_p$ for $\alpha\in (0,1)\cup(1,\infty)$ denotes the following $\alpha$-conditional entropy of a joint distribution $p_{XY}$ given by
\begin{align}  \label{eq:def_CEalpha_classical}
H_{\alpha}^\uparrow(X|Y) &= \frac{\alpha}{1-\alpha}\ln \left( \sum_{y \in \cY} \left( \sum_{x\in \cX} p_{XY}(x,y)^\alpha \right)^{1/\alpha} \right) \\
&= \sup_{q_Y} \left[- D_\alpha( p_{XY} \| \one_X \otimes q_Y )\right]
\end{align}
where the supremum is over probability distributions $q_Y$ on $\cY$, and $D_\alpha$ is the $\alpha$-R\'enyi relative entropy,
\begin{equation}
D_\alpha(p_X \| q_X) = \frac{1}{\alpha-1} \ln \left( \sum_x p_X(x)^{\alpha} q_X(x)^{1-\alpha} \right).
\end{equation}
The second equality of \eqref{eq:def_CEalpha_classical} follows from \cite[Theorem~4]{FB14}.

Arikan's bound \eqref{eq:Arikan-1shot-bound} applies to each $G_{\vec c}(p_{XY},K)$ in \eqref{eq:Ec_q_from_c}, and hence by minimizing over the POVMs $\{E_y\}_{y\in \cY}$, we obtain
\begin{align} \label{eq:1-shot-bounds}
\frac{1}{1 + \ln|\cX|} \exp(H_{\frac{1}{2}}^{\uparrow, M}(X|B)_{\rho}) &\leq  G(X|B)_\rho \\
&\leq \exp(H_{\frac{1}{2}}^{\uparrow, M}(X|B)_\rho),
\end{align}
where for $\alpha \in (0,1)\cup(1,\infty)$, $H_{\alpha}^{\uparrow, M}(X|B)_\rho$ is the $B$-measured conditional $\alpha$-R\'enyi entropy, defined by
\begin{equation}
H_{\alpha}^{\uparrow, M}(X|B)_\rho := \inf_{\{E_y\}_{y \in \cY}} H_{\alpha}^\uparrow(X|Y)_p,
\end{equation}
where $p_{XY}(x,y) = p_X(x) \tr[E_y \rho_B^x]$ is the joint probability distribution obtained by measuring the $B$ part of $\rho_{XB}$ via $\{E_y\}_{y \in \cY}$.

\begin{remark}
We may expand this quantity as
\begin{equation}
H_{\alpha}^{\uparrow, M}(X|B)_\rho = \inf_{\{E_y\}_{y \in \cY}}\sup_{q_Y} \left[ -D_{\alpha}(p_{XY} \| \one_X \otimes q_Y)\right],
\end{equation}
where $p_{XY}$ is induced by the measurement of $\rho_{XB}$.
This quantity seems to be different from the \emph{conditional entropy induced by the measured R\'enyi divergence}, namely
\begin{equation}
 H_{D_\alpha^M}^\uparrow(X|B)_\rho := \sup_{\sigma_B}-D^M_\alpha( \rho_{XB} \| \one_X \otimes \sigma_B),
\end{equation}
where the supremum is over states on the $B$ system, and for any pair of states $(\rho, \sigma)$,
\begin{equation}
D^M_\alpha(\rho\| \sigma) := \sup_{\{E_z\}_{z}} D_\alpha ( \{\tr[E_z \rho]\}_{z} \| \{\tr[E_z \sigma]\}_{z} )
\end{equation}
is the \emph{measured $\alpha$-R\'enyi divergence}. Indeed, the latter quantity may be expanded to obtain
\begin{multline}
 H_{D_\alpha^M}^\uparrow(X|B)_\rho \\
 = \sup_{\sigma_B} \inf_{\{E_z\}_{z}} \left[ -D_\alpha ( \{\tr[E_z \rho_{XB}]\}_{z} \| \{\tr[E_z \one_X \otimes \sigma_B]\}_{z} )\right].
\end{multline}

From the min-max inequality, and the fact that collective measurements on $XB$ can simulate measurements on $B$ alone, we have
\begin{equation}
 H_{D_\alpha^M}^\uparrow(X|B)_\rho  \leq H_{\alpha}^{\uparrow, M}(X|B)_\rho.
\end{equation}
\end{remark}
\subsection{Asymptotic analysis}
We can consider the asymptotic setting in which Bob receives a sequence of product states $\rho_B^{\vec x} := \rho_B^{x_1} \otimes \dotsm \otimes \rho_B^{x_n}$, with probability $p_X(x_1)\dotsm p_X(x_n)$ and aims to guess the full sequence $\vec x = (x_1,\dotsc, x_n)$. In this case, the problem is characterized by the c-q state $\rho_{XB}^{\otimes n}$. The 1-shot bounds \eqref{eq:1-shot-bounds} give us
\begin{align}
 & -\frac{1}{n}\ln\left(1 + n\ln(|\cX|)\right) + \frac{1}{n}H_{\frac{1}{2}}^{\uparrow, M}(X^n | B^n)_{\rho^{\otimes n}} \\ 
 & \qquad \leq   \frac{1}{n}\ln G(X^n|B^n)_{\rho^{\otimes n}} \\
 & \qquad \leq  \frac{1}{n}H_{\frac{1}{2}}^{\uparrow, M}(X^n | B^n)_{\rho^{\otimes n}}
\end{align}
where $H_{\frac{1}{2}}^{\uparrow, M}(X^n | B^n)_{\rho^{\otimes n}}$ can involve collective measurements on the system $B^n$. Taking $n\to \infty$, we obtain
\begin{equation} \label{eq:asymptotic-equality}
     \lim_{n\to\infty} \frac{1}{n}\ln G(X^n|B^n)_{\rho^{\otimes n}} = \lim_{n\to\infty} \frac{1}{n} H_{\frac{1}{2}}^{\uparrow, M}(X^n | B^n)_{\rho^{\otimes n}},
\end{equation}
assuming that the limit on the right-hand side exists.

Note that we can bound
\begin{align}
\frac{1}{n} H_{\frac{1}{2}}^{\uparrow, M}(X^n | B^n)_{\rho^{\otimes n}}
&\leq \frac{1}{n} \inf_{\{ E_{y} \}_{y\in\cY}} H_{\frac{1}{2}}^\uparrow(X^n|Y^n)_{p^{\otimes n}}\\
&=  \inf_{\{ E_{y} \}_{y\in\cY}} H_{\frac{1}{2}}^\uparrow(X|Y)_{p} \\
&= H_{\frac{1}{2}}^{\uparrow, M}(X | B)_{\rho}
\end{align}
where the first inequality follows from the fact that product measurements are a special case of collective measurements, and the first equality follows from the additivity of the classical R\'enyi entropy (\cite[Proposition 1]{Arikan1996}), and the third by the definition of $H_{\frac{1}{2}}^{\uparrow, M}(X | B)_{\rho}$. Moreover, by the data-processing inequality \cite{FL13},
\begin{equation}
\label{eq:relate-ce-to-asymptotic-measured-ce}
\frac{1}{n} H_{\frac{1}{2}}^{\uparrow, M}(X^n | B^n)_{\rho^{\otimes n}} \geq \frac{1}{n} \widetilde H_{\frac{1}{2}}^{\uparrow}(X^n | B^n)_{\rho^{\otimes n}} = \widetilde H_{\frac{1}{2}}^{\uparrow}(X | B)_{\rho},
\end{equation}
where the conditional R\'enyi entropy
$\widetilde H_{\alpha}^{\uparrow}(C | D)_{\sigma}$ of a bipartite state $\sigma_{CD}$ is defined as
\begin{equation}
\widetilde H_{\alpha}^{\uparrow}(C | D)_{\sigma}
= \sup_{\omega_D} \left[-\widetilde D_{\alpha} (\sigma_{CD} \Vert \one_C \otimes \omega_D)\right],
\end{equation}
with the optimization with respect to states $\omega_D$
and the sandwiched R\'enyi relative entropy defined as \cite{MDS+13,WWY14}:
\begin{equation}
    \widetilde D_{\alpha}(X \Vert Y) = \frac{1}{\alpha-1}
    \ln \tr[ (Y^{(1-2\alpha)/\alpha} X Y^{(1-2\alpha)/\alpha})^\alpha].
\end{equation}
The equality in \eqref{eq:relate-ce-to-asymptotic-measured-ce} follows from the additivity of $\widetilde H_{\frac{1}{2}}^{\uparrow}$ under tensor products (see, {{e.g.}}, \cite[Corollary~5.2]{Tom16}). Hence, we obtain
\begin{equation} \label{eq:asymptotic-bounds}
 \widetilde H_{\frac{1}{2}}^{\uparrow}(X | B)_{\rho} \leq \lim_{n\to\infty} \frac{1}{n}\ln G(X^n|B^n)_{\rho^{\otimes n}}  \leq  H_{\frac{1}{2}}^{\uparrow, M}(X | B)_{\rho}.
\end{equation}
In the classical case, \eqref{eq:cq-classical}, both the left and right-hand sides reduce to
\begin{equation}
H_{\frac{1}{2}}^{\uparrow}(X | Y)_{p}
\end{equation}
where $p$ is the underlying classical distribution of \eqref{eq:cq-classical}. Hence, these bounds recover Proposition 5 of~\cite{Arikan1996}.

\section{Semi-definite optimization representations and their consequences} \label{sec:SDP}

The task of calculating $G_{\vec c}(X|B)_\rho$ as defined in \eqref{eq:def_Ec-quantum} can be written as a semi-definite optimization problem, as was found in~\cite{CCWF15}. In this section, we present a different derivation of that fact yielding in \eqref{eq:Ec-SDP} a representation dual to the one found in~\cite{CCWF15}. In \Cref{sec:concavity} we use this representation to prove that the guesswork $G(X|B)_\rho$ is a concave function of the c-q state $\rho_{XB}$. In \Cref{sec:Lipschitz} we likewise use this representation to obtain a Lipschitz continuity bound on the guesswork. Then in \Cref{sec:dual} we compute the dual SDP, recovering the one obtained in \cite{CCWF15}. In \Cref{sec:ub_algo} we use this dual representation to develop a simple algorithm to obtain upper bounds on the quantity. Lastly, in \Cref{sec:misdp}, we formulate a mixed-integer SDP representation of the problem, whose number of variables and constraints scales polynomially with all the relevant quantities (at the cost of adding binary constraints). We also provide implementations of these SDP representations \cite{guesswork_code}, using the Julia programming language \cite{Julia} and the optimization library Convex.jl \cite{Convex.jl-2014}.

Consider an ordered strategy $\vec G$ with a set of POVMs $\{ E_{\vec g} \}_{\vec g \in \cX^K}$. Then since $p_{\vec G, X}(\vec g, x) = p_X(x)\tr[ E_{\vec g} \rho_B^x]$, we have
\begin{equation}
    c_k p_{N(\vec G, X)}(k) = c_k \sum_{x \in \cX} p_X(x)\sum_{\substack{\vec g \in \cX^K\\ N(\vec g, x) = k}} \tr[E_{\vec g} \rho_B^x]
\end{equation}
and hence
\begin{align}
E_{\vec c}(N(\vec G, X)) &= \sum_{k=1}^K c_k \sum_{x \in \cX} p_X(x)\sum_{\substack{\vec g \in \cX^K\\ N(\vec g, x) = k}} \tr[E_{\vec g} \rho_B^x] \\ 
&\qquad + c_\infty  \sum_{x \in \cX} p_X(x)\sum_{\substack{\vec g \in \cX^K\\ N(\vec g, x) = \infty}} \tr[E_{\vec g} \rho_B^x]\\
&= \sum_{\vec g \in \cX^K}\sum_{x \in \cX}  c_{N(\vec g, x)}  p_X(x)\tr[E_{\vec g} \rho_B^x]\\
&= \sum_{\vec g \in \cX^K} \tr[R_{\vec g} E_{\vec g}]
\end{align}
\iffalse
\begin{align}
    E_{\vec c}(N(\vec G, X)) &= \sum_{k=1}^K c_k \sum_{x \in \cX} p_X(x)\sum_{\substack{\vec g \in \cX^K\\ N(\vec g, x) = k}} \tr[E_{\vec g} \rho_B^x] + c_\infty  \sum_{x \in \cX} p_X(x)\sum_{\substack{\vec g \in \cX^K\\ N(\vec g, x) = \infty}} \tr[E_{\vec g} \rho_B^x]\\
&= \sum_{\vec g \in \cX^K}\sum_{x \in \cX}  c_{N(\vec g, x)}  p_X(x)\tr[E_{\vec g} \rho_B^x]\\
&= \sum_{\vec g \in \cX^K} \tr[R_{\vec g} E_{\vec g}]
\end{align}
\fi
where we define $R_{\vec g} := \sum_{x \in \cX}p_X(x) c_{N(\vec g, x)}   \rho_B^x$ for $\vec g\in \cX^K$.
Thus,
 \begin{equation}
 \label{eq:Ec-SDP-full}
 \begin{aligned}
 G_{\vec c}(X|B)_\rho\,\,	=\quad	& \text{minimize}	&	& \sum_{\vec g\in \cX^K}\tr[R_{\vec g} E_{\vec g}]\\
 								& \text{subject to}	&	& E_{\vec g} \geq 0 \qquad \forall \vec g \in \cX^K\\
 								&					&	& \sum_{\vec g \in \cX^K} E_{\vec g} = \one_B.
 \end{aligned}
 \end{equation}
The expression in \eqref{eq:Ec-SDP-full} clarifies that $R_{\vec g}$ has an interpretation as a cost operator corresponding to the guessing outcome $\vec g$.
Since $\sum_{\vec g \in \cX^K}\tr[R_{\vec g} E_{\vec g}]$ is linear in each positive semi-definite (matrix) variable $E_{\vec g}$, $G_{\vec c}(X|B)_\rho$ admits an SDP representation, given in \eqref{eq:Ec-SDP-full}. This program has $|\cX|^K$ variables (each $d_B\times d_B$ complex positive semi-definite matrices), subject to one constraint. Note, however, since the cost vector $\vec c$ is increasing, any guess $\vec h \in \cX^K$ with repeated elements is a suboptimal guessing order. That is, if $\{E_{\vec g}\}_{\vec g \in \cX^K}$ is a POVM with $E_{\vec h} \neq 0$, and
$\vec h' \in \cX^K$ only differs from $\vec h$ by replacing repeated elements such that $\vec h'$ has no repeated elements, then the POVM defined by
\begin{equation}
    \tilde E_{\vec g} := \begin{cases}
    E_{\vec g} & \vec g \neq \vec h \text{ and }\vec g \neq \vec h' \\
    0 & \vec g = \vec h\\
    E_{\vec h} + E_{\vec h'} & \vec g = \vec h'
    \end{cases}
\end{equation}
has $\sum_{\vec g \in \cX^K} \tr[ R_{\vec g} \tilde E_{\vec g}] \leq \sum_{\vec g \in \cX^K} \tr[ R_{\vec g}  E_{\vec g}]$. Hence, we may restrict to the outcome space
\begin{equation}\label{eq:def_X_K}
    \cX^K_{\neq} := \{ \vec g\in \cX^K : g_i \neq g_j, \forall i \neq j \} \subseteq \cX^K.
\end{equation}
Note $|\cX^K_{\neq}| = \frac{|\cX|!}{(|\cX| - K)!}$, and in the case in which $K = |\cX|$, the set $\cX_{K}$ is just the set of permutations of $\cX$. Hence, \eqref{eq:Ec-SDP-full} can be re-written as the following smaller problem (with $\frac{|\cX|!}{(|\cX| - K)!}$ instead of $|\cX|!$ constraints):
 \begin{equation}\label{eq:Ec-SDP}
 \begin{aligned}
 G_{\vec c}(X|B)_\rho\,\,	=\quad	& \text{minimize}	&	& \sum_{\vec g\in \cX^K_{\neq}}\tr[R_{\vec g} E_{\vec g}]\\
 								& \text{subject to}	&	& E_{\vec g} \geq 0 \qquad \forall \vec g \in \cX^K_{\neq}\\
 								&					&	& \sum_{\vec g \in \cX^K_{\neq}} E_{\vec g} = \one_B.
 \end{aligned}
 \end{equation}
Note that in the case $c_\infty = \infty$ and $K < |\cX|$, there exists a finite solution if and only if there exists a POVM $\{ E_{\vec g}\}_{\vec g \in \cX^K_{\neq}}$ such that for all $x\in \cX$ and $\vec g \in \cX^K_{\neq}$ with $x \not \in \vec g$, we have $\tr[ E_{\vec g} \rho_B^x] = 0$. Whether or not this holds depends on the particular state $\rho_{XB}$.
However, when $c_\infty < \infty$ or $K = |\cX|$, for any state $\rho_{XB}$, the problem \eqref{eq:Ec-SDP}  has a finite solution. Moreover, for any POVM $\{ E_{\vec g}\}_{\vec g \in \cX^K_{\neq}}$, the objective $\sum_{\vec g\in \cX^K_{\neq}}\tr[R_{\vec g} E_{\vec g}]$ is finite. In the following, we restrict to those two cases.

\begin{remark}
This optimization problem has the same form as that of discriminating quantum states in an ensemble, as described in, {{e.g.}}, \cite[Section 3.2.1]{Wat18}. 
Note, however, that (1) the $R_{\vec g}$ are positive semi-definite but not normalized, and (2) the case of having two copies of the unknown state in the guessing framework does not correspond to $R_{\vec g}^{\otimes 2}$.  Nevertheless,
slight modifications to~\cite[Theorem 3.9]{Wat18} show that a POVM $\{E_{\vec g}\}_{\vec g \in \cX^K_{\neq}}$ is optimal for \eqref{eq:Ec-SDP} if and only if
\begin{equation}
Y = \sum_{\vec g \in \cX^K_{\neq}} R_{\vec g} E_{\vec g} 
\end{equation}
satisfies $Y \leq R_{\vec g}$ for all $\vec g \in \cX^K_{\neq}$.
\end{remark}

\begin{remark}
The set of POVMs is convex and since the objective function is linear, any minimizer for \eqref{eq:Ec-SDP} may be decomposed into extremal POVMs which are also minimizers. By~\cite[Corollary 2.2]{Par99}, any extremal POVM on a Hilbert space of size $d_B$ has at most $d_B^2$ non-zero elements. Hence, there exist minimizers of \eqref{eq:Ec-SDP} with at most $d_B^2$ non-zero elements (even though $|\cX^K_{\neq}|$ could be far larger than $d_B^2$). Let $S \subseteq \cX^K_{\neq}$ be a set of $d_B^2$ points such that there exists $\{\tilde E_{\vec g}\}_{\vec g \in S}$ with $\tilde E_{\vec g}\geq 0$, $\sum_{\vec g \in S} \tilde E_{\vec g} = \one_B$, and
\begin{equation}
G_{\vec c}(X|B)_\rho = \sum_{\vec g \in S} \tr[\tilde E_{\vec g} R_{\vec g}].
\end{equation}
Then \eqref{eq:Ec-SDP} holds with $\cX^K_{\neq}$ replaced by $S$, namely
 \begin{equation}\label{eq:Ec-SDP-small}
 \begin{aligned}
 G_{\vec c}(X|B)_\rho\,\,	=\quad	& \text{minimize}	&	& \sum_{\vec g\in S}\tr[R_{\vec g} E_{\vec g}]\\
 								& \text{subject to}	&	& E_{\vec g} \geq 0 \qquad \forall \vec g \in S\\
 								&					&	& \sum_{\vec g \in S} E_{\vec g} = \one_B.
 \end{aligned}
 \end{equation}
 Note the ``$\leq$'' direction of the equality \eqref{eq:Ec-SDP-small} is trivial, since given a minimizer $\{E_{\vec g}\}_{\vec g \in S}$ for \eqref{eq:Ec-SDP-small}, simply extending it by choosing $E_{\vec g} = 0$ for $\vec g \not \in S$ gives a feasible point for the optimization problem on the right-hand side of \eqref{eq:Ec-SDP}. The ``$\geq$'' direction follows from the existence of the $\{\tilde E_{\vec g}\}_{\vec g \in S}$ described above. We note here that it is nontrivial to identify the set $S$. In general, identifying the set $S$ is as difficult as solving the original problem in \eqref{eq:Ec-SDP}. However, applying a heuristic method inspired by choosing a smaller set of constraints can lead to useful upper bounds on the guesswork, which we describe in Section~\ref{sec:ub_algo}.
\end{remark}

\subsection{The dual problem}\label{sec:dual}

Next, we compute the dual problem to \eqref{eq:Ec-SDP}, in the case $K= |\cX|$ or $c_\infty < \infty$. Consider the Lagrangian
\begin{align}	
 &\cL((E_{\vec g})_{\vec g \in \cX^K_{\neq}}, (\lambda_{\vec g})_{\vec g \in \cX^K_{\neq}}, \nu) = \sum_{\vec g\in \cX^K_{\neq}}\braket{ R_{\vec g}, E_{\vec g}} \\
 & \qquad- \sum_{\vec g\in \cX^K_{\neq}}\braket{\lambda_{\vec g}, E_{\vec g}} + \Braket{\nu, \sum_{\vec g \in \cX^K_{\neq}} E_{\vec g} - \one_B}\\
 &= \sum_{\vec g\in \cX^K_{\neq}}\braket{R_{\vec g}-\lambda_{\vec g} + \nu, E_{\vec g}}  - \tr[\nu]
\end{align}
where we have introduced the Hilbert--Schmidt product $\braket{A,B} = \tr[A^\dagger B]$, and where $\lambda_{\vec g}\geq 0$ is the dual variable to the inequality constraint $E_{\vec g}\geq 0$, and $\nu = \nu^\dagger$ is the dual variable to the equality constraint $\sum_{\vec g \in \cX^K_{\neq}} E_{\vec g} = \one_B$. As shown in, {{e.g.}}, \cite{BV04}, the primal problem \eqref{eq:Ec-SDP} can be expressed as
\begin{equation}\label{eq:primal-formula}
 \min_{(E_{\vec g})_{\vec g \in \cX^K_{\neq}}}\max_{\lambda_{\vec g} \geq 0, \nu}\cL((E_{\vec g})_{\vec g \in \cX^K_{\neq}}, (\lambda_{\vec g})_{\vec g \in \cX^K_{\neq}}, \nu)
\end{equation}
while the dual problem is given by
\begin{equation}\label{eq:dual-formula}
\max_{\lambda_{\vec g} \geq 0, \nu}\min_{(E_{\vec g})_{\vec g \in \cX^K_{\neq}}}\cL((E_{\vec g})_{\vec g \in \cX^K_{\neq}}, (\lambda_{\vec g})_{\vec g \in \cX^K_{\neq}}, \nu).
\end{equation}
If $R_{\vec g} - \lambda_{\vec g} + \nu \neq 0$ for any $\vec g\in \cX^K_{\neq}$, then the inner minimization in \eqref{eq:dual-formula} yields $-\infty$. Hence,
\begin{multline}
\min_{(E_{\vec g})_{\vec g \in \cX^K_{\neq}}}\cL((E_{\vec g})_{\vec g \in \cX^K_{\neq}}, (\lambda_{\vec g})_{\vec g \in \cX^K_{\neq}}, \nu) \\
= \begin{cases}
- \infty & R_{\vec g} - \lambda_{\vec g} + \nu \neq 0 \quad \exists \vec g \in \cX^K_{\neq}\\
- \tr[\nu] & \text{else.}
\end{cases}
\end{multline}
\iffalse
\begin{equation}
    \min_{(E_{\vec g})_{\vec g \in \cX^K_{\neq}}}\cL((E_{\vec g})_{\vec g \in \cX^K_{\neq}}, (\lambda_{\vec g})_{\vec g \in \cX^K_{\neq}}, \nu) = 
    \begin{cases}
- \infty & R_{\vec g} - \lambda_{\vec g} + \nu \neq 0 \quad \exists \vec g \in \cX^K_{\neq}\\
- \tr[\nu] & \text{else.}
\end{cases}
\end{equation}
\fi
The constraint $\lambda_{\vec g} \geq 0$ and $R_{\vec g} - \lambda_{\vec g} + \nu = 0$ imply the semi-definite inequality $-\nu \leq R_{\vec g}$. Writing $Y = - \nu$ and maximizing over $\lambda_{\vec g} \geq 0$, \eqref{eq:dual-formula} becomes
\begin{equation} \label{eq:dual_SDP}
\begin{aligned}
\text{maximize} \quad & \tr[Y]\\
\text{subject to} \quad & Y = Y^\dagger \\
& Y \leq R_{\vec g} \qquad \forall \vec g \in \cX^K_{\neq}
\end{aligned}
\end{equation}
Since \eqref{eq:Ec-SDP} is strictly feasible ({{e.g.}}, $E_{\vec g} = \one_B \frac{1}{|\cX^K_{\neq}|}$ is a strictly feasible point) by Slater's condition, strong duality holds. Hence, \eqref{eq:dual_SDP} obtains the same optimal value as \eqref{eq:Ec-SDP}. The formulation of the problem as given in  \eqref{eq:dual_SDP} was previously found in the work~\cite[Proposition 3]{CCWF15}. In contrast to the primal SDP \eqref{eq:Ec-SDP-small}, the dual problem has a single variable $Y$ subject to $| \cX^K_{\neq} | $ constraints.

\subsection{A simple algorithm to compute upper bounds}\label{sec:ub_algo}

The dual form of the SDP can be used to generate upper bounds on $G_{\vec c}(X|B)_\rho$ simply by removing constraints. This provides an algorithm to find an upper bound on the objective function: Decide on some number of constraints $\kappa$ to impose in total. Then,
\begin{enumerate}
	\item  Initialize an empty list $L=\{\}$ corresponding to constraints to impose.
	\item Set $Y$ to be the identity matrix, as a first guess at the optimal dual variable.
	\item If $Y$ satisfies $Y \leq R_{\vec g}$ for all $\vec g \in \cX^K_{\neq}$, then $Y$ is the maximizer of the dual problem \eqref{eq:dual_SDP}, and the optimization is solved. Otherwise, find $\vec g\in \cX^K_{\neq}$ such that $Y \not \leq R_{\vec g}$, and add $\vec g$ to the list $L$.
	\item  Solve the problem
	\begin{equation} \label{eq:dual_SDP-L}
	\begin{aligned}
	\text{maximize} \quad & \tr[Y]\\
	\text{subject to} \quad & Y = Y^\dagger \\
	& Y \leq R_{\vec g} \qquad \forall \vec g \in L
	\end{aligned}
	\end{equation}
	and set $Y$ to be its maximizer.
	\item Repeat steps 2 and 3 until the list $L$ has length $\kappa$.
	\item Solve the problem one last time, and return the output.
\end{enumerate}
In order to find a constraint that $Y$ violates, a heuristic technique such as simulated annealing can be used. Moreover, in the case that there are too many constraints to fit into memory or check exhaustively, using an iterative technique (such as simulated annealing) is essential. If this algorithm was continued (without imposing a limit on the total number of constraints $\kappa$ to impose), it would eventually yield the true value $G_{\vec c}(\rho_{XB},K)$. When a total number of constraints is limited, it provides an upper bound (since it is a relaxation of \eqref{eq:dual_SDP}).

However, even with a limit $\kappa$ on the total number of constraints, this algorithm can in theory yield the true value $G_{\vec c}(X|B)_\rho$. Note that the dual problem to \eqref{eq:Ec-SDP-small} is 
\begin{equation} \label{eq:dual_SDP-S}
\begin{aligned}
\text{maximize} \quad & \tr[Y]\\
\text{subject to} \quad & Y = Y^\dagger \\
& Y \leq R_{\vec g} \qquad \forall \vec g \in S
\end{aligned}
\end{equation}
where $S \subseteq \cX^K_{\neq}$ has $|S| = d_B^2$ and is described in the remark above. Hence, if $L$ in \eqref{eq:dual_SDP-L} equals $S$, then the algorithm finds the true value $G_{\vec c}(X|B)_\rho$, not just an upper bound. Thus, $\kappa = d_B^2$ suffices if the constraints $\vec g$ can be chosen precisely to obtain $L=S$. In general, finding $S$ is as difficult as solving the original problem. Nonetheless, this motivates why choosing a relatively small value of $\kappa$ (such as $d_B^2$) can still yield good upper bounds.

\subsection{A mixed-integer reformulation}\label{sec:misdp}

The problem of computing $G_{\vec c} (X|B)_{\rho}$ can be formulated another way as a mixed-integer SDP, \textit{i.e.}, an SDP that has additional integer or binary constraints. Consider a POVM $\{F_j\}_{j=1}^M$ with $M$ outcomes. When outcome $j$ is obtained, Bob guesses in some order $\vec g^{(j)} \in \cX^K_{\neq}$. Then consider the problem
\begin{equation} \label{eq:misdp-1}
\begin{aligned}
	& \text{minimize}	&	& \sum_{x\in \cX, j = 1,\dotsc, M}p_X(x)  c_{N(\vec g^{(j)}, x)}   \tr[ F_j \rho_B^x]\\
 								& \text{subject to}	&	& F_j \geq 0 \qquad j = 1,\dotsc, M,\\
 								&					&	& \vec g^{(j)} \in \cX^K_{\neq}, j = 1,\dotsc, M,\\
 								&					&	& \sum_{j=1}^M F_j = \one_B.
\end{aligned}
\end{equation}
We note that in the above, the variables to be optimized over are both the POVM $\{F_j\}_{j=1}^M$ and the guessing orders $\{\vec g^{(j)} \}_j$ corresponding to each POVM outcome.
This optimization is not an SDP, since the dependence on the optimization variables $\{\vec g^{(j)}\}_{j=1}^M$ and $\{F_j\}$ is not linear, and $\vec g^{(j)} \in \cX^K_{\neq}$ is a discrete constraint. Consider, however, the case that $K = |\cX|$. With this assumption, we will be able to remove the nonlinearity although not the discrete variables. This yields a \emph{mixed-integer} SDP: an optimization problem such that if all integer constraints were removed, the result would be an SDP. We proceed as follows.

Under the condition $K = |\cX|$, we may restrict to considering guessing orders that are permutations without loss of generality;  other guessing orders have repeated guesses, which can only increase the value of the objective function. In this case, the outcome $\infty$ never occurs, and for each $\vec g \in S_{|\cX|}$, the quantity $( c_{N(\vec g, x)} )_{x \in \cX}$ satisfies
\begin{equation}
( c_{N(\vec g, x)} )_{x \in \cX} = \vec g \, \inv(c),
\end{equation}
where $\vec g \inv$ is the inverse permutation to $\vec g$, and $c = (c_k)_{k=1}^{K}$ is the cost vector (without $\infty$). Here, $S_n$ is the set of permutations on $\{1,\dotsc,n\}$. Let $P^{(j)}$ be an $|\cX|\times |\cX|$ matrix representation of the permutation $\vec g^{(j)}{}\inv$. Then $(P^{(j)} c)_x = \sum_{y\in \cX} P^{(j)}_{xy} c_y = c_{N(\vec g, x)}$. Hence, the optimization \eqref{eq:misdp-1} can be reformulated as
\begin{equation}\label{eq:misdp-2}
\begin{aligned}
	& \text{minimize}	&	&  \sum_{x,y\in \cX, j = 1,\dotsc, M}p_X(x)  P^{(j)}_{xy} c_y   \tr[ F_j \rho_B^x]\\
 								& \text{subject to}	&	& F_j \in \cM_{d_B} \qquad \forall\,j \in [M],\\
 								&					&	& P^{(j)}_{xy} \in \{0,1\}, \quad \forall\,j \in [M], x,y \in \cX\\
 								& 	&	& F_j \geq 0 \qquad \forall\,j \in [M],\\
 								&					&	& \sum_{j=1}^M F_j = \one_B,\\
 								&					&	& \sum_{x\in \cX}P^{(j)}_{xy} = 1, \quad \forall\,j \in [M], y \in \cX\\
 								&					&	& \sum_{y\in \cX}P^{(j)}_{xy} = 1, \quad \forall\,j \in [M], x \in \cX\\
 					\end{aligned}
\end{equation}
Note that all the constraints are semi-definite or linear, except that each element $P^{(j)}_{xy}$ is a binary variable: $P^{(j)}_{xy} \in \{0,1\}$, which is a particularly simple type of discrete constraint. The non-linearity in the objective function, however, persists. To remove this, we take advantage of the  fact that the  $P^{(j)}_{xy}$ are binary. In particular,~\cite[Equations (22)--(24)]{BDNK19} provide  a clever trick to turn objective functions with terms of the form $z x$ where $z$ is a binary variable and $x$ a continuous variable into objective functions of a continuous variable $y$ subject to four affine constraints (in terms of $x$ and $z$), as long as $x$ is bounded by known constants. We reproduce this argument in the following.

 We first write the objective function entirely in terms of scalar quantities:
\begin{multline}
\sum_{x,y\in \cX, j \in [M]}p_X(x)  P^{(j)}_{xy} c_y   \tr[ F_j \rho_B^x] \\
= \sum_{k,\ell \in [d_B]}\sum_{x,y\in \cX, j \in [M]}p_X(x)(\rho_B^x)_{k\ell}  c_y\,  P^{(j)}_{xy}  (F_j)_{\ell k} 
\end{multline}
\iffalse
\begin{equation}
    \sum_{x,y\in \cX, j \in [M]}p_X(x)  P^{(j)}_{xy} c_y   \tr[ F_j \rho_B^x] = \sum_{k,\ell \in [d_B]}\sum_{x,y\in \cX, j \in [M]}p_X(x)(\rho_B^x)_{k\ell}  c_y\,  P^{(j)}_{xy}  (F_j)_{\ell k} 
\end{equation}
\fi
Let $x = (F_j)_{\ell k} $ and $z = P^{(j)}_{xy} \in \{0,1\}$. Then $|x| \leq \tr[F_j]/2 \leq d_B/2$. Then $x_L := - d_B/2$ and $x_U := d_B/2$ constitute lower and upper bounds to $x$, respectively. Hence, the following four inequalities hold trivially:
\begin{equation}
\begin{aligned}\label{eq:ref1}
z (x - x_L) \geq 0,\\
(z-1)(x - x_U) \geq 0,\\
z(x - x_U) \leq 0,\\
(z-1)(x - x_L) \leq 0.
\end{aligned}
\end{equation}
Now, let $y = xz$. Then we have
\begin{equation}\label{eq:ref2}
\begin{aligned}
y - zx_L \geq 0,\\
y - z x_U \geq x - x_U,\\
y - z x_U \leq 0,\\
y - zx_L \leq x - x_L.
\end{aligned}
\end{equation}
On the other hand, let us remove the constraint $y= xz$, and consider $y$ as another variable. Then if $z = 0$, the first equation of \eqref{eq:ref2} implies that $y \geq 0$, while the third implies $y\leq 0$, so $y= 0$. On the other hand, if $z=1$, then the second equation of \eqref{eq:ref2} implies that $y \geq x$ while the fourth implies that $y \leq x$. Hence, either way, $y = xz$. Thus, \eqref{eq:ref2} is equivalent to $y = xz$.

With this transformation, \eqref{eq:misdp-2} can be reformulated as the following. 
\begin{equation} \label{eq:misdp-final}
\begin{aligned}
	& \text{minimize}	&	&  \sum_{k,\ell \in [d_B]}\sum_{x,y\in \cX, j \in [M]}p_X(x)(\rho_B^x)_{k\ell}  c_y\, y_{xy \ell k j}\\
 								& \text{subject to}	&	& F_j \in \cM_{d_B} \qquad \forall\,j \in [M],\\
 								&					&	& y_{xy \ell k j} \in \mathbb{R}, \quad  \forall \,x,y \in \cX, \ell,k \in [d_B], j \in [M],\\
 								&					&	& P^{(j)}_{xy} \in \{0,1\}, \quad \forall\,j \in [M], x,y \in \cX\\
 								& 	&	& F_j \geq 0 \qquad \forall\,j \in [M],\\
 								&					&	& \sum_{j=1}^M F_j = \one_B,\\
 								&					&	& \sum_{x\in \cX}P^{(j)}_{xy} = 1 \quad \forall\,j \in [M], y \in \cX, \\
 								&					&	& \sum_{y\in \cX}P^{(j)}_{xy} = 1 \quad \forall\,j \in [M], x \in \cX, \\
&					&	&  y_{xy \ell k j} + P^{(j)}_{xy} \frac{d_B}{2} \geq 0,\\
&					&	&  y_{xy \ell k j} - P^{(j)}_{xy}  \frac{d_B}{2} \geq (F_j)_{\ell k}  - \frac{d_B}{2},\\
&					&	&  y_{xy \ell k j} - P^{(j)}_{xy}  \frac{d_B}{2} \leq 0,\\
&					&	&  y_{xy \ell k j} + P^{(j)}_{xy} \frac{d_B}{2} \leq (F_j)_{\ell k}  + \frac{d_B}{2}.
\end{aligned}
\end{equation}
where the last four constraints hold for $\forall \,x,y \in \cX, \ell,k \in [d_B], \text{ and } j \in [M]$. This is a mixed-integer SDP, with a number of constraints and variables that is polynomial in $M, d_B, |\cX|$. Moreover, if $M \geq d_B^2$, then as follows from the remark below \eqref{eq:Ec-SDP}, the mixed-integer SDP \eqref{eq:misdp-final} obtains the same optimal value as \eqref{eq:Ec-SDP}, namely $G_{\vec c}(\rho_{XB}, |\cX|)$, using that $K = |\cX|$. Note, however, that mixed-integer SDPs are not in general efficiently solvable; they encompass mixed integer linear programs, which are NP-hard. However, in practice they can sometimes be quickly solved. Since the original SDP formulation \eqref{eq:Ec-SDP} involves an exponential (in $|\cX|$) number of variables (or an exponential number of constraints in its dual formulation \eqref{eq:dual_SDP}), Eq.~\eqref{eq:misdp-final}  may provide a more practical approach in some cases because it instead has a polynomial (in $|\cX|$) number of variables. Mixed-integer SDPs can be solved in various ways; in the code accompanying this paper \cite{guesswork_code}, the problem \eqref{eq:misdp-final} is solved using the library Pajarito.jl \cite{CLV18}, which proceeds by solving an alternating sequence of mixed-integer linear problems and SDPs.

\section{The ellipsoid algorithm}

The ellipsoid algorithm (see, e.g., \cite{GLS93}) provides a theoretical proof that under a strict feasibility assumption, semi-definite programs can be solved in time that scales as a polynomial in: the number of scalar variables and constraints, the logarithm of a 2-norm bound on the feasible points, $\ln(1/\eps)$ where $\eps$ is the solution tolerance, and the maximum bit length of the scalar entries of the objective and constraints (see e.g. \cite[Theorem 4]{Wat09a}). 

In fact, the ellipsoid algorithm applies quite generally to the optimization of a linear objective function over a convex \emph{feasible region} (which could be described by a domain and constraint functions). The ellipsoid algorithm only requires a \emph{separation oracle} for the feasible region, a subroutine which either asserts that a given point lies within the feasible region, or provides a separating hyperplane between the given point and the feasible region. When the separation oracle can be evaluated in polynomial time, the overall ellipsoid algorithm runs in polynomial time as well (see \cite[Corollary 4.2.7]{GLS93}).

In the case of a single positive semi-definite constraint, e.g. $Y \geq 0$, a simple separation oracle is given by computing the eigendecomposition of $Y$ and checking if all of its eigenvalues are non-negative. If so, it returns that $Y$ is indeed feasible, and otherwise returns the matrix $C := U \text{diag}(f(\lambda_1),\dotsc,f(\lambda_d)) U^\dagger$ where $Y = U\text{diag}(\lambda_1,\dotsc,\lambda_d)U^\dagger$ is the eigendecomposition of $Y$, $U$ is unitary, $\lambda_1,\dotsc,\lambda_d$ are the eigenvalues, and $f(x) = 1$ if $x < 0$ and $f(x) = 0$ otherwise. This matrix has the properties that $C\geq 0$, $\|C\|_\infty = 1$, and $\tr[C^\dagger Y] = \sum_{i=1}^d f(\lambda_i)\lambda_i = \sum_{i: \lambda_i < 0}\lambda_i <0$.

In the case of the dual problem \eqref{eq:dual_SDP} with $K=|\cX|$, we have $|\cX|!$ positive semi-definite constraints. Thus, we cannot check all of them together in polynomial time.

The feasibility problem $Y\leq R_{\pi}$ for each $\pi \in S_{|\cX|}$ can be written as the following mixed-integer non-linear problem,
\begin{equation*}	
\begin{aligned}
\eta := \text{minimize} &\quad \braket{\psi, \left(\sum_{x\in \cX}(P c)_{x} p_X(x) \rho_B^x - Y\right) \psi}\\
\text{subject to} &\quad \sum_{i\in \cX} P_{ij} = \sum_{j\in \cX} P_{ij} = 1\\
&\quad P_{ij}\in \{0\}, \,\,i,j \in \cX \\
&\quad \psi \in \bC^{d_B}
\end{aligned}
\end{equation*}
where $\eta \geq 0$ if and only if $Y \leq R_{\pi}$ for all $\pi \in S_{|\cX|}$, using that a matrix $M$ satisfies $M\geq 0$ if and only if $\braket{\psi, M \psi}\geq 0$ for all $\psi \in \bC^{d_B}$. Since the convex hull of the permutation matrices is given by the doubly stochastic matrices, the discrete constraints can be relaxed, yielding the following reformulation
\begin{equation}\label{eq:nl_feasibility}
\begin{aligned}
\eta = \text{minimize} &\quad \braket{\psi, \left(\sum_{x\in \cX}(Dc)_{x} p_X(x) \rho_B^x - Y\right) \psi}\\
\text{subject to} &\quad \sum_{i\in \cX} D_{ij} = \sum_{j\in \cX} D_{ij} = 1\\
&\quad D_{ij} \geq 0, \,\,i,j \in \cX \\
&\quad \psi \in \bC^{d_B}.
\end{aligned}
\end{equation}
If $\eta\geq 0$, then $Y$ is feasible. Otherwise, the optimal value $D^*$ can be decomposed as a convex combination of permutations, $D^* = \sum_i \alpha_i P_i$, and we must have $\sum_{x} (P_i c)_x p_X(x) \rho_B^x \not \geq Y$ for some $P_i$, using that the objective is an affine function of $D$. The problem \eqref{eq:nl_feasibility} can be solved by global non-linear optimization solvers such as EAGO.jl \cite{WS17} or SCIP \cite{GamrathEtal2020OO}, but not in general in polynomial time.

At each iteration of the ellipsoid algorithm, one must evaluate the separation oracle for some Hermitian matrix $Y$. In order to avoid solving \eqref{eq:nl_feasibility}, one may attempt to prove the feasibility or infeasibility of a point $Y$ by other means. For example, one may search over permutations $\pi$ heuristically, in order to find $R_\pi$ such that $Y\not\leq R_\pi$. If such a permutation can be identified, then $Y$ is not feasible, and the problem \eqref{eq:nl_feasibility} does not need to be solved. Likewise, if one can show that for some $k \in (1,\dotsc,|\cX|)$,
\begin{multline}
Y \leq \sum_{i=1}^k c_{|\cX|-i} p_{x_{i}} \rho_B^{x_{i}} + c_1 \sum_{x \in \cX \setminus \{x_{i_1},\dotsc,x_{i_k}\}}p_x \rho_B^x \\
\forall (x_1,\dotsc,x_k)\in \cX^k_{\neq}
\end{multline}
\iffalse
\begin{equation}
    Y \leq \sum_{i=1}^k c_{|\cX|-i} p_{x_{i}} \rho_B^{x_{i}} + c_1 \sum_{x \in \cX \setminus \{x_{i_1},\dotsc,x_{i_k}\}}p_x \rho_B^x \forall (x_1,\dotsc,x_k)\in \cX^k_{\neq}
\end{equation}
\fi
then $Y$ must be feasible, and again \eqref{eq:nl_feasibility} does not need to be solved. The number of comparisons required scales as $|\cX|^k$; for small choices of $k$, this provides an efficient check for feasibility (which may, however, be inconclusive).

\subsection{Numerical comparisons}

We compare numerical implementations of several of the above algorithms on a set of 12 test problems. The code for these experiments can be found at \cite{guesswork_code}. Each problem has $p \equiv u$, the uniform distribution $u := (1/|\cX|, \dotsc, 1/|\cX|)$, for simplicity. The states are chosen as
\begin{enumerate}
	\item Two random qubit density matrices
	\item Two random qutrit density matrices
	\item Three pure qubits chosen equidistant within one  plane of the Bloch sphere (the qubit trine states), i.e.\@
	\[
	\cos\left(j \tfrac{2\pi}{3}\right) \ket{0} + \sin\left(j \tfrac{2\pi}{3}\right) \ket{1}, \qquad j = 1,2,3
	\]
	\item Three random qubit density matrices
	\item Three random qutrit density matrices
	\item The four BB84 states, $\ket{0}, \ket{1}$, and $\ket{\pm} = \frac{1}{\sqrt{2}}(\ket{0} \pm \ket{1})$
\end{enumerate}
as well as the ``tensor-2'' case of
\begin{equation}\label{eq:tensor-2}
\{ \rho \otimes \sigma : \rho, \sigma \in S\}
\end{equation}for each of the six sets $S$ listed above, corresponding to the guesswork problem with quantum side information associated to $\rho_{XB}^{\otimes 2}$, where $\rho_{XB}$ is the state associated to the original guesswork problem with quantum side information. The random states were chosen uniformly at random (i.e.\@~according to the Haar measure).

The exponentially-large SDP formulation (and its dual), the mixed-integer SDP algorithm, and the active set method were compared, with several choices of parameters and underlying solvers. The mixed-integer SDP formulation was evaluated with $M=d_B$ (yielding an upper bound), $M=d_B^2$ (yielding the optimal value), with the Pajarito mixed-integer SDP solver \cite{CLV18}, using Convex.jl (version 0.12.7) \cite{Convex.jl-2014} to formulate the problem. Pajarito proceeds by solving mixed-integer linear problems (MILP) and SDPs as subproblems, and thus uses both a MILP solver and an SDP solver as subcomponents. Pajarito provides two algorithms: an iterative algorithm, which alternates between solving MILP and SDP subproblems, and solving a single branch-and-cut problem in which SDP subproblems are solved via so-called lazy callbacks to add cuts to the mixed-integer problem. The latter is called ``mixed-solver drives'' (MSD) in the Pajarito documentation. We tested three configurations of Pajarito (version 0.7.0):

\begin{description}
\item[(c1)] Gurobi (version 9.0.3) as the MILP solver and MOSEK
(version 8.1.0.82) as the SDP solver, with Pajarito's MSD algorithm,
\item[(c2)] Gurobi as the MILP solver and MOSEK as the SDP solver, with Pajarito's iterative algorithm, with a relative optimality gap tolerance of $0$,
\item[(o)] Cbc \cite{Cbc} (version 2.10.3) as the MILP solver, and SCS \cite{SCS} (version 2.1.1) as the SDP solver, with Pajarito's iterative algorithm.
\end{description}

Here, `c' stands for commercial, and `o' for open-source. In the configuration (c1), Gurobi was set to have  a relative optimality gap tolerance of $10^{-5}$ and in (c2), a relative optimality gap tolerance of $0$. In both configurations, Gurobi was given an absolute linear-constraint-wise feasibility tolerance of $10^{-8}$, and an integrality tolerance of $10^{-9}$. These choices of parameters match those made in \cite{CLV18}. Cbc was given an integrality tolerance of $10^{-8}$, and SCS's (normalized) primal, dual residual and relative gap were set to $10^{-6}$ for each problem. The default parameters were used otherwise. Note the MSD option was not used with Cbc, since the solver does not support lazy callbacks.

For the (exponentially large) SDP primal and dual formulations, the problems were solved with both MOSEK and SCS, and likewise with the active-set upper bound.

The active set method uses simulated annealing to iteratively add violated constraints to the problem to find an upper bound, as described in \Cref{sec:ub_algo}, and uses a maximum-time parameter $t_\text{max}$ to stop iterating when the estimated time of finding another constraint to add would cause the running time to exceed the maximum-time\footnote{The maximum time can still be exceeded, since at least one iteration must be performed and the estimate can be wrong.}. This provides a way to compare the improvement (or lack thereof) of running the algorithm for more iterations. The algorithm also terminates when a violated constraint cannot be found after 50 runs of simulated annealing (started each time with different random initial conditions). Here, the problems were solved with three choices of $t_\text{max}$, $20$\,s, $60$\,s, and $240$\,s.

The exact answer was not known analytically for most of these problems, and so the average relative error was calculated by comparing to the mean of the solutions (excluding the active-set method and the MISDP with $M=d_B$, which only give an upper bound in general). For the cases of the BB84 states and the Y-states, where the solution is known exactly (see \Cref{sec:example}), the solutions obtained here match the analytic value to a relative tolerance of at least $10^{-7}$.

The problems were run sequentially on a four-core desktop computer (Intel i7-6700K 4.00GHz CPU, with 16 GB of RAM, on Ubuntu-20.04 via Windows Subsystem for Linux, version 2), via the programming language Julia \cite{Julia} (version 1.5.1), with a 5 minute time limit. The results are summarized in \Cref{tab:summary}, and presented in more detail in \Cref{tab:first_problems} and \Cref{tab:second_problems}.

One can see that the MISDP problems were harder to solve than the corresponding SDPs for these relatively small problem instances. The MISDPs have the advantage of finding extremal solutions, however, in the case $M=d_B$, and may scale better to large instances. Additionally, the active-set upper bound performed fairly well, finding feasible points within 20\,\% of the optimum in all cases, with only $t_\text{max}=20$\,s, and often finding near-optimal solutions. It was also the only method able to scale to the largest instances tested, such as two copies of the BB84 states (which involves 16 quantum states in dimension four, and for which the SDP formulation has 16!\@ variables.). In general, the commercial solvers performed better than the open source solvers, with the notable exception of the active-set upper bound with MOSEK, in which two more problems timed out than with SCS. This could be due to SCS being a first-order solver which can therefore possibly scale to larger problem instances than MOSEK, which is a second-order solver.

\begin{table*}[t!]
\begin{adjustbox}{center}
\begin{tabular}{p{2.75cm}p{2.25cm}p{1.75cm}p{1.75cm}p{1.75cm}p{1.75cm}p{1.75cm}}
Algorithm & Parameters & average rel. error & average time & number solved & number timed out & number errored out\\ \toprule
MISDP ($d_B$) & Pajarito (c1) & 0\,\% & 23.74\,s & 6 & 6 & 0\\
 & Pajarito (c2) & 0\,\% & 24.03\,s & 6 & 6 & 0\\
 & Pajarito (o) & 0\,\% & 45.05\,s & 6 & 6 & 0\\\addlinespace[0.75em]
MISDP ($d_B^2$) & Pajarito (c1) & 0\,\% & 35.27\,s & 5 & 7 & 0\\
 & Pajarito (c2) & 0\,\% & 27.97\,s & 4 & 8 & 0\\
 & Pajarito (o) & 0\,\% & 131.35\,s & 4 & 8 & 0\\\addlinespace[0.75em]
SDP & MOSEK & 0\,\% & 8.99\,s & 8 & 3 & 1\\
 & SCS & 0\,\% & 9.08\,s & 8 & 3 & 1\\\addlinespace[0.75em]
SDP (dual) & MOSEK & 0\,\% & 8.74\,s & 8 & 3 & 1\\
 & SCS & 0\,\% & 8.59\,s & 8 & 3 & 1\\\addlinespace[0.75em]
\multirow[t]{3}{2.75cm}[0pt]{Active set upper bound (MOSEK)} & $t_\text{max}=20$\,s & 6.80\,\% & 16.08\,s & 10 & 2 & 0\\
 & $t_\text{max}=60$\,s & 6.79\,\% & 19.03\,s & 10 & 2 & 0\\
 & $t_\text{max}=240$\,s & 6.80\,\% & 26.17\,s & 10 & 2 & 0\\\addlinespace[0.75em]
\multirow[t]{3}{2.75cm}[0pt]{Active set upper bound (SCS)} & $t_\text{max}=20$\,s & 6.09\,\% & 33.24\,s & 12 & 0 & 0\\
 & $t_\text{max}=60$\,s & 6.09\,\% & 35.78\,s & 12 & 0 & 0\\
 & $t_\text{max}=240$\,s & 6.09\,\% & 34.30\,s & 11 & 1 & 0
\end{tabular}
\end{adjustbox}
\vspace{1em}
\caption{\label{tab:summary} Comparison of average relative error and average solve time for the 12 problems discussed above. A problem is considered ``timed out'' if an answer is not obtained in 5 minutes, and ``errored out'' if the solution was not obtained due to errors (such as running out of RAM). The average relative error, which was rounded to two decimal digits, and the time taken are calculated only over the problems which were solved by the given algorithm and choice of parameters. ``MISDP ($d_B$)'' refers to the choice $M=d_B$, and likewise ``MISDP ($d_B^2$)'' refers to the choice $M=d_B^2$.}
\end{table*}

\bigskip

\begin{table*}[t!]
\begin{adjustbox}{center}
\begin{tabular}{p{2.75cm}p{2.25cm}p{3cm}@{\hspace{1cm}}p{3cm}@{\hspace{1cm}}p{3cm}}
Algorithm & Parameters & Two\mbox{ }random qubits & Two\mbox{ }random qutrits & Y-states\\ \toprule
MISDP ($d_B$) & Pajarito (c1) & 23.63\,s, timeout & 23.60\,s, timeout & 23.56\,s, timeout\\
& Pajarito (c2) & 22.99\,s, timeout & 23.31\,s, timeout & 23.21\,s, timeout\\
& Pajarito (o) & 23.47\,s, timeout & 24.77\,s, timeout & 26.15\,s, timeout\\\addlinespace[0.75em]
MISDP ($d_B^2$) & Pajarito (c1) & 24.49\,s (0.00\,\%), timeout & 31.40\,s (0.00\,\%), timeout & 26.27\,s (0.00\,\%), timeout\\
 & Pajarito (c2) & 25.02\,s (0.00\,\%), timeout & 31.39\,s (0.00\,\%), timeout & 29.08\,s (0.00\,\%), timeout\\
 & Pajarito (o) & 26.54\,s (0.00\,\%), timeout & 212.79\,s (0.00\,\%), timeout & 141.14\,s (0.00\,\%), timeout\\\addlinespace[0.75em]
SDP & MOSEK & 8.69\,s, 8.84\,s & 8.78\,s, 9.23\,s & 9.33\,s, timeout\\
& SCS & 9.00\,s, 8.90\,s & 8.44\,s, 11.22\,s & 8.98\,s, timeout\\\addlinespace[0.75em]
SDP (dual) & MOSEK & 8.46\,s, 8.63\,s & 8.54\,s, 8.83\,s & 9.11\,s, timeout\\
 & SCS & 8.76\,s, 8.32\,s & 8.33\,s, 9.20\,s & 8.74\,s, timeout\\\addlinespace[0.75em]
\multirow[t]{3}{2.75cm}[0pt]{Active set upper bound (MOSEK)} & $t_\text{max}=20$\,s & 8.76\,s\mbox{ }(19.5\,\%), 10.41\,s\mbox{ }(1.5\,\%) & 8.89\,s\mbox{ }(19.5\,\%), timeout & 9.72\,s\mbox{ }(0\,\%), 34.25\,s\mbox{ }(?\@\,\%)\\
 & $t_\text{max}=60$\,s & 10.91\,s\mbox{ }(19.5\,\%), 10.41\,s\mbox{ }(1.9\,\%) & 8.87\,s\mbox{ }(19.5\,\%), timeout & 9.66\,s\mbox{ }(0\,\%), 31.00\,s\mbox{ }(?\@\,\%)\\
 & $t_\text{max}=240$\,s & 9.47\,s\mbox{ }(19.5\,\%), 10.40\,s\mbox{ }(1.5\,\%) & 8.90\,s\mbox{ }(19.5\,\%), timeout & 9.81\,s\mbox{ }(0\,\%), 30.26\,s\mbox{ }(?\@\,\%)\\\addlinespace[0.75em]
\multirow[t]{3}{2.75cm}[0pt]{Active set upper bound (SCS)} & $t_\text{max}=20$\,s & 9.04\,s\mbox{ }(19.5\,\%), 10.92\,s\mbox{ }(1.5\,\%) & 8.70\,s\mbox{ }(19.5\,\%), 101.06\,s\mbox{ }(1.1\,\%) & 9.23\,s\mbox{ }(0\,\%), 82.84\,s\mbox{ }(?\@\,\%)\\
 & $t_\text{max}=60$\,s & 9.07\,s\mbox{ }(19.5\,\%), 10.22\,s\mbox{ }(1.9\,\%) & 8.66\,s\mbox{ }(19.5\,\%), 32.94\,s\mbox{ }(1.2\,\%) & 9.18\,s\mbox{ }(0\,\%), 50.79\,s\mbox{ }(?\@\,\%) \\
 & $t_\text{max}=240$\,s & 9.04\,s\mbox{ }(19.5\,\%), 10.02\,s\mbox{ }(1.9\,\%) & 8.79\,s\mbox{ }(19.5\,\%), 22.69\,s\mbox{ }(1.2\,\%) & 9.31\,s\mbox{ }(0\,\%), 37.36\,s\mbox{ }(?\@\,\%)
\end{tabular}
\end{adjustbox}

\bigskip

\caption{\label{tab:first_problems} The individual timings for each algorithm and choice of settings on problems (1)--(3), and the corresponding ``tensor-2'' problems discussed at \eqref{eq:tensor-2}. For each algorithm, the running time of the original problem is given followed by the running time on the ``tensor-2'' problem, e.g.\@ the SDP formulation with MOSEK on the two random qubits problem was solved in 8.69 seconds, and in 8.84 seconds for the corresponding tensor-2 problem. ``timeout'' is written whenever the problem was not solved within 5 minutes. For the active set algorithms, the relative error is also given for each problem in parenthesis. Note that the MISDP formulation with $M=d_B$ is also only known to be an upper bound, but a relative error of less than $10^{-5}$ in each instance, so the relative errors are omitted. Lastly, the relative error is written as {?\@\,\%} in the case that only an upper bound was obtained.}
\end{table*}

\begin{table*}[t!]
\begin{adjustbox}{center}
\begin{tabular}{p{2.75cm}p{2.25cm}p{3cm}@{\hspace{1cm}}p{3cm}@{\hspace{1cm}}p{3cm}}
Algorithm & Parameters & Three\mbox{ }random qubits & Three\mbox{ }random qutrits & BB84\mbox{ }states\\ \toprule
MISDP ($d_B$) & Pajarito (c1) & 23.63\,s, timeout & 24.49\,s, timeout & 23.51\,s, timeout\\
 & Pajarito (c2) & 23.17\,s, timeout & 26.50\,s, timeout & 25.01\,s, timeout\\
 & Pajarito (o) & 27.11\,s, timeout & 95.01\,s, timeout & 73.80\,s, timeout\\\addlinespace[0.75em]
MISDP ($d_B^2$) & Pajarito (c1) & 25.82\,s (0.00\,\%), timeout & timeout, timeout & 68.35\,s (0.00\,\%), timeout\\
& Pajarito (c2) & 26.38\,s (0.00\,\%), timeout & timeout, timeout & timeout, timeout\\
& Pajarito (o) & 144.93\,s (0.00\,\%), timeout & timeout, timeout & timeout, timeout\\\addlinespace[0.75em]
SDP & MOSEK & 9.35\,s, timeout & 8.82\,s, timeout & 8.87\,s, error\\
 & SCS & 9.09\,s, timeout & 8.46\,s, timeout & 8.54\,s, error\\ \addlinespace[0.75em]
SDP (dual) & MOSEK & 9.14\,s, timeout & 8.55\,s, timeout & 8.62\,s, error\\
& SCS & 8.86\,s, timeout & 8.25\,s, timeout & 8.23\,s, error\\\addlinespace[0.75em]
\multirow[t]{3}{2.75cm}[0pt]{Active set upper bound (MOSEK)}  & $t_\text{max}=20$\,s & 9.73\,s\mbox{ }(1.3\,\%), 32.29\,s\mbox{ }(?\@\,\%) & 9.50\,s\mbox{ }(5.8\,\%), timeout & 9.51\,s\mbox{ }(0\,\%), 27.72\,s\mbox{ }(?\@\,\%)\\
& $t_\text{max}=60$\,s & 9.69\,s\mbox{ }(1.3\,\%), 24.12\,s\mbox{ }(?\@\,\%) & 9.45\,s\mbox{ }(5.8\,\%), timeout & 9.77\,s\mbox{ }(0\,\%), 66.47\,s\mbox{ }(?\@\,\%)\\
& $t_\text{max}=240$\,s & 9.69\,s\mbox{ }(1.3\,\%), 30.42\,s\mbox{ }(?\@\,\%) & 9.45\,s\mbox{ }(5.8\,\%), timeout & 9.49\,s\mbox{ }(0\,\%), 133.81\,s\mbox{ }(?\@\,\%)\\\addlinespace[0.75em]
\multirow[t]{3}{2.75cm}[0pt]{Active set upper bound (SCS)}  & $t_\text{max}=20$\,s & 9.44\,s\mbox{ }(1.3\,\%), 35.66\,s\mbox{ }(?\@\,\%) & 9.67\,s\mbox{ }(5.8\,\%), 84.43\,s\mbox{ }(?\@\,\%) & 9.09\,s\mbox{ }(0\,\%), 28.76\,s\mbox{ }(?\@\,\%)\\
& $t_\text{max}=60$\,s & 9.44\,s\mbox{ }(1.3\,\%), 54.60\,s\mbox{ }(?\@\,\%) & 9.11\,s\mbox{ }(5.8\,\%), 155.51\,s\mbox{ }(?\@\,\%) & 9.29\,s\mbox{ }(0\,\%), 70.57\,s\mbox{ }(?\@\,\%)\\
& $t_\text{max}=240$\,s & 8.89\,s\mbox{ }(1.3\,\%), 75.43\,s\mbox{ }(?\@\,\%) & 9.01\,s\mbox{ }(5.8\,\%), timeout & 9.15\,s\mbox{ }(0\,\%), 177.64\,s\mbox{ }(?\@\,\%)
\end{tabular}
\end{adjustbox}

\bigskip

\caption{\label{tab:second_problems} The individual timings for each algorithm and choice of settings on problems (4)--(6). See \Cref{tab:first_problems} for a description of the quantities shown. Here, ``error'' means the solution was not obtained due to an error (such as running out of memory).}
\end{table*}

\section{Properties of guesswork with quantum side information}

\subsection{Concavity of the guesswork}\label{sec:concavity}
\begin{proposition}\label{prop:guesswork_concave}
For each cost vector $\vec c$ and $K \leq |\cX|$, the function
\begin{equation} \label{eq:Guesswork_fn_rho}
 \rho_{XB} \mapsto G_{\vec c}(X|B)_\rho
\end{equation}
from the set of c-q states of the form \eqref{eq:cq-state} to $\mathbb{R}_{\geq 0}\cup \{\infty\}$, is concave. 
\end{proposition}

\begin{proof}
For $\vec g \in \cX^K_{\neq}$, and $\rho_{XB}$ a c-q state, the quantity $R_{\vec g}^\rho :=  \sum_{x\in \cX}p_X(x)  c_{N(\vec g, x)} \rho_B^x$ can be expressed as
\begin{equation}
R_{\vec g}^\rho = \tr_X\!\left[\left(\sum_{x\in \cX}c_{N(\vec g, x)}  |x\rangle \!\langle x|_X \otimes I_B \right)\rho_{XB}\right]
\label{eq:alt-form-R-g}
\end{equation}
and hence is linear in $\rho_{XB}$.
Then for each POVM $(E_{\vec g})_{\vec g \in \cX^K_{\neq}}$, 
\begin{equation}
\rho_{XB} \mapsto \sum_{\vec g \in \cX^K_{\neq}} \tr[ R_{\vec g}^\rho E_{\vec g}]
\end{equation}
is linear in $\rho_{XB}$. The arbitrary infimum of concave functions, and in particular linear functions, is concave, and hence 
\begin{equation}
 G_{\vec c}(X|B)_\rho\equiv \min_{(E_{\vec g})_{\vec g \in \cX^K_{\neq}}} \sum_{\vec g \in \cX^K_{\neq}} \tr[R_{\vec g}^\rho E_{\vec g}],
 \end{equation}
 where the minimum is taken over all POVMs on system $B$ with outcomes in $\cX^K_{\neq}$,
 is concave.
\end{proof}
\begin{remark}
\Cref{prop:guesswork_concave} carries over to guesswork without side information, $G(X)$, which simply corresponds to the case that $\rho_B^x \equiv \rho_B$ is independent of $x \in \cX$. Since $G(X)$ is manifestly symmetric under permutations of the density $p_X$, this proves that $G(X)$ is a \emph{Schur concave} function of the distribution $p_X$ ({i.e.}, decreasing in the majorization pre-order; see, e.g., \cite{MOA11} for an overview of majorization and Schur concave functions). Consequently, the work \cite{HD18} provides an algorithm to calculate local continuity bounds for $G(X)$.
\end{remark}

\subsection{Continuity of the guesswork}\label{sec:Lipschitz}

\begin{proposition}\label{prop:guesswork_Lipschitz}
For each cost vector $\vec c$ and $K \leq |\cX|$, such that either $c_\infty < \infty$ or $K = |\cX|$, the function
\begin{equation}
 \rho_{XB} \mapsto G_{\vec c}(X|B)_\rho
\end{equation}
from the set of c-q states of the form \eqref{eq:cq-state} to $\mathbb{R}_{\geq 0}$, is Lipschitz continuous, satisfying the bound
\begin{equation}
|G_{\vec c}(X|B)_\rho - G_{\vec c}(X|B)_\sigma| \leq \kappa \|\rho_{XB}-\sigma_{XB}\|_1.
\end{equation}
for any c-q states $\rho_{XB}$ and $\sigma_{XB}$, where $\kappa = c_\infty$ if $K < |\cX|$, and $\kappa = c_{|\cX|}$ if $K = |\cX|$.
\end{proposition}
\begin{proof}
Define
\begin{equation}
f(\rho_{XB}, \{E_{\vec g}\}_{\vec g \in \cX^K_{\neq}}) :=  \sum_{\vec g \in \cX^K_{\neq}} \tr[ R_{\vec g}^\rho E_{\vec g}].
\end{equation}
Then, by linearity (as discussed in the proof of \Cref{prop:guesswork_concave}),
\begin{align}
& f(\rho_{XB}, \{E_{\vec g}\}_{\vec g \in \cX^K_{\neq}})-f(\sigma_{XB}, \{E_{\vec g}\}_{\vec g \in \cX^K_{\neq}}) \notag \\
& \qquad = \sum_{\vec g \in \cX^K_{\neq}} \tr[ R_{\vec g}^{\rho-\sigma} E_{\vec g}]\\
& \qquad = \sum_{\vec g \in \cX^K_{\neq}} \tr[ \tr_X[C^{(\vec g)}_{XB} \Delta_{XB}] E_{\vec g}]
\end{align}
using \eqref{eq:alt-form-R-g}, where $C_{XB}^{(\vec g)} := \sum_{x\in \cX}c_{N(\vec g, x)}  |x\rangle \!\langle x| \otimes I_B\geq 0$ and $\Delta_{XB} := \rho_{XB}-\sigma_{XB}$. Since $C_{XB}^{(\vec g)}$ and $\Delta_{XB}$ commute, using the c-q structure of each, we have
\begin{align}
    & f(\rho_{XB}, \{E_{\vec g}\}_{\vec g \in \cX^K_{\neq}})-f(\sigma_{XB}, \{E_{\vec g}\}_{\vec g \in \cX^K_{\neq}})\notag \\
    & \qquad =\sum_{\vec g \in \cX^K_{\neq}} \tr[ C^{(\vec g)}_{XB} \Delta_{XB} (I_X \otimes E_{\vec g})] \\
    & \qquad = \tr\!\left[ \Delta_{XB}\sum_{\vec g \in \cX^K_{\neq}}C^{(\vec g)}_{XB} (I_X\otimes E_{\vec g})\right].
\end{align}
Set
\iffalse
\begin{multline}
    F_{XB} := \sum_{\vec g \in \cX^K_{\neq}}C^{(\vec g)}_{XB} (I_X\otimes E_{\vec g}) \\
    =\sum_{x\in \cX}\sum_{\vec g \in \cX^K_{\neq}}  c_{N(\vec g, x)}  |x\rangle \!\langle x| \otimes E_{\vec g}.
\end{multline}
\fi
\begin{equation}
     F_{XB} := \sum_{\vec g \in \cX^K_{\neq}}C^{(\vec g)}_{XB} (I_X\otimes E_{\vec g})  =\!\!\sum_{x\in \cX}\sum_{\vec g \in \cX^K_{\neq}}  c_{N(\vec g, x)}  |x\rangle \!\langle x| \otimes E_{\vec g}.
\end{equation}
Since $c_{N(\vec g, x)}\leq \kappa$ for each $x \in \cX$ and $\vec g \in \cX^K_{\neq}$, we have that $F_{XB} \leq \kappa \sum_{x\in \cX}\sum_{\vec g \in \cX^K_{\neq}}  |x\rangle \!\langle x| \otimes E_{\vec g} $ in semi-definite order. Performing the sums, we have $F_{XB} \leq \kappa \, I_X \otimes I_B$ and hence $\|F_{XB}\|_\infty \leq \kappa$. Thus,

\begin{align}
    & f(\rho_{XB}, \{E_{\vec g}\}_{\vec g \in \cX^K_{\neq}})-f(\sigma_{XB}, \{E_{\vec g}\}_{\vec g \in \cX^K_{\neq}}) \notag \\
    & \qquad = \tr\!\left[ \Delta_{XB} F_{XB}\right]\\
    & \qquad \leq \|  \Delta_{XB} F_{XB}\|_1\\
    & \qquad \leq \|\Delta_{XB}\|_1 \,\|F_{XB}\|_\infty\\
    & \qquad \leq \kappa \|\rho_{XB}-\sigma_{XB}\|_1
\end{align}
using H\"older's inequality in the second to last inequality.
Swapping $\rho_{XB}$ and $\sigma_{XB}$ completes the proof.
\end{proof}

\section{Simple examples}\label{sec:example}
% \FloatBarrier

\subsection{BB84 states as side information}

\begin{figure} % side caption figure from the package `sidecap`
    \centering
    \includegraphics[width=0.48\textwidth]{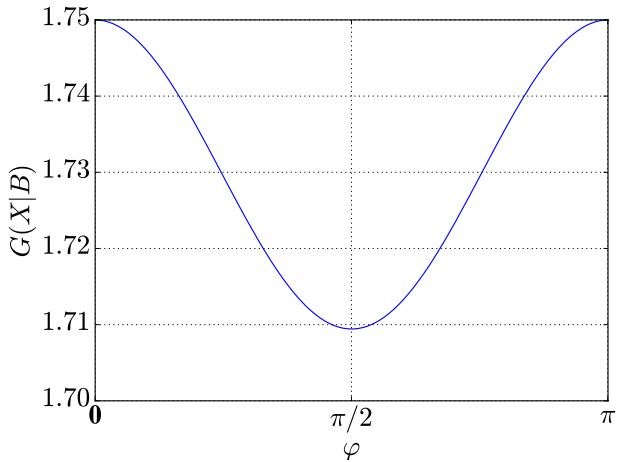}
    \caption{The guesswork $G(X|B)$ as a function of the parameter $\varphi$, when $X$ is uniformly distributed over $\{1,2,3,4\}$ and the corresponding side information states are $\{ \ket{0}, \ket{1}, \ket{\psi(\varphi)}, \ket{\psi(-\varphi)}\}$, where $\ket{\psi(\varphi)} = \cos (\varphi/2) \ket{0} + \sin (\varphi/2) \ket{1}$. We see that for classical states, i.e., when $\varphi = 0$ or $\varphi = \pi$, we obtain a maximum value of 1.75. For the BB84 states ($\varphi = \pi/2$), we achieve a minimum.}
    \label{fig:bb84-type-general-plot}
\end{figure}

As an example, we consider the problem of calculating guesswork when one has four uniformly distributed letters to guess from, each correlated to one of the four BB84 states~\cite{BB84}. That is,
\begin{equation}
    \rho_{XB} = \frac{1}{4} \sum_{k = 1}^4 |x_k\rangle \!\langle x_k|_X \otimes |\psi_k\rangle \!\langle \psi_k|_B
\end{equation}
with the four $\ket{\psi_k}$'s being chosen from $\left\{ \ket{0}, \ket{1}, \ket{+}, \ket{-} \right\}$. This example is firmly in the quantum realm of guesswork, as more information about the side information system $B$ can be obtained via a quantum measurement than a classical one (in the computational basis, that is). 

% Using our SDP formulation as in \eqref{eq:Ec-SDP}, we compute the expected guesswork for this example on a personal computer as 1.709430. Additionally, below we derive the explicit protocol for achieving this value of expected guesswork.
We establish an analytic upper bound on the guesswork by considering a particular POVM and associated sequences of guesses. We consider the POVM consisting of two orthogonal projectors $|\theta\rangle \!\langle \theta|$ and $|\theta^{\perp}\rangle \!\langle \theta^{\perp}|$ with $\ket{\theta} := \sin \theta \ket{0} + \cos \theta \ket{1}$. If the outcome corresponding to $|\theta\rangle \!\langle \theta|$ is obtained, then we guess in the order corresponding to $(1, +, -, 0)$. Similarly, the guessing order corresponding to the other outcome $\theta^{\perp}$ is $(0, -, +, 1)$, which is the reverse order.

By calculating the objective function of \eqref{eq:Ec-SDP}, we obtain
\begin{align} \label{eq:bb84-guesswork}
G(X|B) &\leq \begin{multlined}[t][5cm] \frac{1}{2} \Big( 1 \cdot \cos^2 \theta + 2 \cdot \frac{1}{2} \left(1 + \sin 2 \theta \right) +  \\ 3 \cdot \frac{1}{2} \left( 1 - \sin 2 \theta \right) + 4 \cdot \sin^2 \theta \Big) \end{multlined} \\
					   &= 1.75 + \frac{3}{2} \sin^2 \theta - \frac{1}{4} \sin 2 \theta.
\end{align}
With the aim of minimizing the guesswork, we choose $\theta = \frac{1}{2} \arctan{\frac{1}{3}}$, and obtain the right-hand side of \eqref{eq:bb84-guesswork} as $\frac{1}{4} \left( 10 - \sqrt{10} \right) \approx 1.709430$. Moreover, the SDP in \eqref{eq:Ec-SDP} can be solved numerically to obtain the same value, providing a matching numerical lower bound; see \cite{guesswork_code} for the code involved, including a high-precision demonstration using the SDP solver SDPA-GMP \cite{YFN+10} showing agreement to 200 digits.

We also consider a generalization of this example, where the side information states are chosen from the set $\{ \ket{0}, \ket{1}, \ket{\psi(\varphi)}, \ket{\psi(-\varphi)}\}$ where $\ket{\psi(\varphi)} = \cos (\varphi/2) \ket{0} + \sin (\varphi/2) \ket{1}$. The BB84 states are a special case of this ensemble with $\varphi = \pi/2$. For each of these ensembles, we compute the guesswork using our SDP formulation in \eqref{eq:Ec-SDP}. The results are shown in Figure~\ref{fig:bb84-type-general-plot}.

Furthermore, we can use this example to delineate the difference, or gap, between guesswork with classical and quantum information. There are two ways of reducing the example of BB84 states to a classical setting: (a) restricting to measurements in the standard basis $\{ \ket{0}, \ket{1} \}$, or (b) replacing the side information states $\ket{+}$ and $\ket{-}$ with the maximally mixed qubit state $\pi_2$. In both of these cases, the side information then takes the form of a random variable $Y$. The joint probability distribution of variables $XY$ is given as follows:

\begin{table}[H]
\centering
\begin{tabular}{| l | l | l |}
\hline
\diagbox{$X$}{$Y$} & 0  & 1  \\
\hline
                    $x_1$   & 1/4 & 0  \\
                     $x_2$  & 1/8 & 1/8  \\
                     $x_3$  & 1/8 & 1/8  \\
                     $x_4$  & 0 & 1/4 \\
                     \hline
\end{tabular}
\end{table}

Given either value of $Y$, one needs an average of $1.75$ guesses. Hence in this classical analogue of the BB84 states, the guesswork is $1.75$. This is higher than the lower value ($\approx 1.709$) that can be achieved by quantum measurements.

\subsection{Qubit trine states as side information}

In this example, we consider the problem of calculating guesswork with three uniformly distributed letters to guess from, each correlated to one of the three qubit trine states \cite{Holevo1973}. That is,
\begin{equation}
    \rho_{XB} = \frac{1}{3} \sum_{k = 1}^3 |x_k\rangle \!\langle x_k|_X \otimes |\psi_k\rangle \!\langle \psi_k|_B
\end{equation}
where $\ket{\psi_k} = \cos\left(k \tfrac{2\pi}{3}\right) \ket{0} + \sin\left(k \tfrac{2\pi}{3}\right) \ket{1}$.

As in the previous example of the BB84 states, here too we establish an analytic upper bound on the guesswork by considering a particular POVM. Consider that the measurement is characterized by two orthogonal projectors $|\theta\rangle \!\langle \theta|$ and $|\theta^{\perp}\rangle \!\langle \theta^{\perp}|$ with $\ket{\theta} := \cos \theta \ket{0} + \sin \theta \ket{1}$. For the sake of simplicity, we restrict to $\theta \in [0, \pi/2]$. For the outcome corresponding to $\theta$, we guess in the order corresponding to $(\psi_2, \psi_3, \psi_1)$, and for the outcome corresponding to $\theta^\perp$, we guess in the order corresponding to $(\psi_1, \psi_3, \psi_2)$.

The objective function in \eqref{eq:Ec-SDP} leads us to
\begin{equation} \label{eq:ystate-guesswork}
\begin{split}
G(X|B) &\leq \frac{1}{2} \cdot \frac{2}{3} \Big( 1 \cdot \left(\cos^2 (\theta - 4\pi/3)  + \sin^2 (\theta - 2\pi/3) \right) \\ & \qquad +  2\cdot \left( \cos^2 \theta + \sin^2 \theta \right) \\ & \qquad +  3 \cdot  \left( \cos^2 (\theta - 2\pi/3)  + \sin^2 (\theta - 4\pi/3) \right) \Big) \\
					   &= \frac{4}{3} + \frac{2}{3} \left( \cos^2 (\theta - 2\pi/3) + \sin^2 (\theta - 4\pi/3)  \right).
\end{split}
\end{equation}

The guesswork $G(X|B)$ is minimized by setting $f'(\theta) = 0$ where $f(\theta) =  \cos^2 (\theta - 2\pi/3) + \sin^2 (\theta - 4\pi/3)$. This leads to $\theta = \pi/4$ and $G(X|B) = (2 - 1/\sqrt{3}) \approx 1.422649$. The SDP \eqref{eq:Ec-SDP} is solved for this example as well, and the numerical result shows agreement with the analytic upper bound up to a relative tolerance of at least $10^{-7}$.

Further, in this example, if we restrict to measuring in the standard basis (mimicking the classical analogue of the side information), then the average number of guesses needed is $1.5$, in contrast to $\approx1.4227$ guesses needed using the optimal quantum measurement.

\section{Guesswork as a security criterion: certifying an imperfect key state}\label{sec:security}

A primitive in any cryptography scheme is the establishment of a secret key between two communicating parties. Quantum key distribution (QKD) protocols can produce a certifiably secure secret key by using pre-shared entanglement \cite{Ekert:1991:661}. However, if the protocol is not implemented perfectly, as is the case in realistic scenarios, then some information can leak out to an eavesdropper. How secure is the key obtained in this ``imperfect'' scenario? In other words, if there is a small deviation from the ideal protocol, how does it affect the security of  the key? We address this question considering the guesswork as a security criterion.

Consider two systems $K$ and $E$, where $K$ denotes the key system and encodes the secret key, and $E$ is the system held by the eavesdropper. An ideal key state is of the form $\pi_K \otimes \rho_E$ where $\pi_K$ refers to the maximally mixed state on the key system. This means that the eavesdropper can learn nothing about the key with access to the $E$ system alone. An imperfect key, generally, is the joint state $\rho_{KE}$. Consider the promise that the imperfect key state is $\varepsilon$-close to an ideal one in normalized trace distance:
\begin{equation} \label{eq:promise}
    \frac{1}{2} \Vert \rho_{KE} - \pi_K \otimes \rho_E \Vert_1 \leq \varepsilon. 
\end{equation}
For an ideal key state, the expected guesswork for the eavesdropper is $\sum_{k} \frac{1}{|K|} k = \frac{|K|+1}{2}$ where $|K|$ indicates the cardinality of the alphabet associated to system $K$. We have the following result:
\begin{theorem}\label{thm:imperfect-key}
For an imperfect key state satisfying the promise in \eqref{eq:promise},  the following bound on guesswork holds
\begin{equation}\label{eq:robustness-bound}
    G(K|E) \geq \frac{|K| + 1 }{2}  -  |K| \varepsilon.
\end{equation}
\end{theorem}

\begin{proof}
We apply the result of Lemma~\ref{lem:pliam-side-info} below, which holds for the case of guesswork with classical side information.
We know from Theorem \ref{thm:ordered-from-adaptive} that a measured strategy for guesswork is equivalent to a quantum strategy. Using that fact, and by combining the promise $\frac{1}{2} \Vert \rho_{KE} - \pi_K \otimes \rho_E \Vert_1 \leq \varepsilon$ and the result in Lemma \ref{lem:pliam-side-info}, we have \eqref{eq:robustness-bound}.
\end{proof}

\Cref{thm:imperfect-key} provides a robustness guarantee that imperfect key states continue to have near-maximal guesswork, if they remain close to an ideal key state in normalized trace distance.

Our proof of the lower bound in \eqref{eq:robustness-bound}, as given above, is a consequence of the following extension of an analogous result pertaining to guesswork, due to Pliam~\cite[Theorem~3]{Pliam1998}. Pliam's inequality states that for any random variable $X$ with probability distribution $p_X$, 
\begin{equation} \label{eq:pliam}
	\frac{| \mathcal{X} |+1}{2} - G(X) \leq \frac{1}{2} | \mathcal{X} | \, \Vert p_X - u_X \Vert_1,
\end{equation}
where $G(X)$ denotes the guesswork and $u_X$ denotes the uniform distribution.

\begin{lemma} \label{lem:pliam-side-info}
For random variables $X$ and $Y$, the following bound holds for the guesswork:
\begin{equation} 
\frac{\left\vert \mathcal{X}\right\vert +1}{2}-G(X|Y)\leq\frac{\left\vert
\mathcal{X}\right\vert }{2}\left\Vert p_{XY}-u_{X}\otimes p_{Y}\right\Vert
_{1}.
\end{equation}
\end{lemma}

\begin{proof}
Consider the case of a joint distribution $p_{XY}$, with conditional distribution $p_{X|Y}$ and marginal distribution $p_{Y}$, and suppose that the value of $y$ is fixed. Then we can invoke Pliam's bound \eqref{eq:pliam} to find that%
\begin{equation}
\frac{\left\vert \mathcal{X}\right\vert +1}{2}-G(X|Y=y)\leq\frac{\left\vert
\mathcal{X}\right\vert }{2}\left\Vert p_{X|Y=y}-u_{X}\right\Vert _{1},
\end{equation}
where the notation $G(X|Y=y)$ indicates the guesswork (without side information) of a random variable distributed according to $p_{X|Y}(\cdot|y)$. Taking the expectation of both sides with respect to the random variable$~Y$,
we find that%
\begin{align}
&\frac{\left\vert \mathcal{X}\right\vert +1}{2}-\sum_{y}p_{Y}(y)G(X|Y=y) \\
& \qquad \leq\frac{\left\vert \mathcal{X}\right\vert }{2}\sum_{y}p_{Y}(y)\left\Vert
p_{X|Y=y}-u_{X}\right\Vert _{1}\\
& \qquad =\frac{\left\vert \mathcal{X}\right\vert }{2}\sum_{y}p_{Y}(y)\sum
_{x}\left\vert p_{X|Y}(x|y)-u_{X}(x)\right\vert \\
& \qquad =\frac{\left\vert \mathcal{X}\right\vert }{2}\sum_{y}\sum_{x}\left\vert
p_{X|Y}(x|y)p_{Y}(y)-u_{X}(x)p_{Y}(y)\right\vert \\
&\qquad  =\frac{\left\vert \mathcal{X}\right\vert }{2}\sum_{y}\sum_{x}\left\vert
p_{XY}(x,y)-u_{X}(x)p_{Y}(y)\right\vert \\
& \qquad =\frac{\left\vert \mathcal{X}\right\vert }{2}\left\Vert p_{XY}-u_{X}\otimes
p_{Y}\right\Vert _{1}.
\end{align}
Using the fact that%
\begin{equation}
\sum_{y}p_{Y}(y)G(X|Y=y)=G(X|Y),
\end{equation}
we can conclude the generalization of \eqref{eq:pliam} in the presence of classical side information%
\begin{equation} 
\frac{\left\vert \mathcal{X}\right\vert +1}{2}-G(X|Y)\leq\frac{\left\vert
\mathcal{X}\right\vert }{2}\left\Vert p_{XY}-u_{X}\otimes p_{Y}\right\Vert
_{1}.
\end{equation}
This concludes the proof.
\end{proof}

% \begin{equation}
% 	G(X|E) \geq \frac{| \mathcal{X} |+1}{2} - | \mathcal{X} | \varepsilon.
% \end{equation}

\begin{remark}
Note that \Cref{prop:guesswork_Lipschitz} gives the following continuity bound for the guesswork near  $\pi_K \otimes \rho_E$:
\begin{align}
    |G(K|E)_{\rho} - G(K|E)_{\pi\otimes \rho_E}| \leq 2 \eps |K| ,
\end{align}
 and hence
\begin{align}
    G(K|E)_{\rho} \geq %|K|\left(\frac{1}{2}-2\eps \right)+\frac{1}{2} = 
    \frac{|K|+1}{2} - 2 |K| \eps.
\end{align}
Thus, the bound in \eqref{eq:robustness-bound} is slightly better than what we obtain by employing \Cref{prop:guesswork_Lipschitz}.
\end{remark}

\section{Open questions}

Guesswork presents an operationally-relevant method to quantify uncertainty, and has been relatively unexplored in the presence of quantum side information. We hope our investigation opens the door to further analysis of the guesswork and methods to compute it. In particular, our work leaves open the following questions:
\begin{enumerate}
	\item Does equality hold in \eqref{eq:relate-ce-to-asymptotic-measured-ce}? If so, the single-letter expression
	\begin{equation}
	    \lim_{n\to\infty} \frac{1}{n}\ln G(X^n|B^n)_{\rho^{\otimes n}} = \widetilde H_{\frac{1}{2}}^{\uparrow}(X | B)_{\rho}
	\end{equation}
	holds, matching the classical case \cite[Prop. 5]{Arikan1996}.
	\item Ref.~\cite{BFT17} presented variational expressions for the measured R\'enyi divegerences $D^M_\alpha$ and showed how those lead to efficient ways to compute the divergences. Are there similar variational formulas for $H^{\uparrow, M}_\alpha(X|Y)_\rho$? That could similarly provide an efficient way to compute the quantity.
\end{enumerate}

\section*{Acknowledgements} E.H. would like to thank Harsha Nagarajan for pointing out the transformation in \cite[Equations (22)--(24)]{BDNK19}. E.H.~is supported by the Cantab Capital Institute for the Mathematics of Information (CCIMI). V.K.~acknowledges support from the Louisiana State University Economic Development Assistantship. M.M.W.~acknowledges support from the US National Science Foundation through grant no.~1907615.

\bibliographystyle{IEEEtran}
\bibliography{guesswork}
\end{document}